\documentclass[runningheads]{llncs}
\usepackage{latexsym} \usepackage{amssymb} \usepackage{amstext}
\usepackage{amsmath}
\usepackage{wrapfig}
\usepackage{mathtools}
\usepackage{graphicx}
\usepackage{tabularx}
\usepackage{array}
\newcolumntype{L}[1]{>{\raggedright\let\newline\\\arraybackslash\hspace{0pt}}m{#1}}
\newcolumntype{C}[1]{>{\centering\let\newline\\\arraybackslash\hspace{0pt}}m{#1}}
\newcolumntype{R}[1]{>{\raggedleft\let\newline\\\arraybackslash\hspace{0pt}}m{#1}}
\usepackage[english]{babel}
\usepackage{xspace}
\usepackage{hyperref}
\usepackage{todonotes}
\usepackage{paralist}
\usepackage{caption}
\usepackage{subcaption}
\usepackage{wrapfig}
\newcommand{\myi}{(\emph{i})\xspace}
\newcommand{\myii}{(\emph{ii})\xspace}

 \renewcommand{\L}{\mathcal{L}}
 
 \renewcommand{\P}{\mathcal{P}}

\newcommand{\U}{\mathcal{U}}

\newcommand{\Next}{\raisebox{-0.27ex}{\LARGE$\circ$}}

\newcommand{\Until}{\mathop{\U}}

\newcommand{\true}{\mathit{true}}

\newcommand{\false}{\mathit{false}}

\newcommand{\LTL}{{\sc ltl}\xspace}

\newcommand{\Nat}{{\rm I\kern-.23em N}}
\newcommand{\Prop}{Prop}

\newcommand{\temptrue}{\mathit{temp\_true}}
\newcommand{\tempfalse}{\mathit{temp\_false}}

\newcommand{\ra}{\rightarrow}

\newcommand{\msf}{\mathsf}
\newcommand{\ol}{\overline}

\newcommand{\bpm}{{\sc bpm}\xspace}
\newcommand{\bpmn}{{\sc bpmn}\xspace}
\newcommand{\wsp}{{\sc wsp}\xspace}

\newcommand{\sod}{{\sc s}o{\sc d}\xspace}
\newcommand{\bod}{{\sc b}o{\sc d}\xspace}

\newcommand{\rv}{\textsc{rv}}

\newcommand{\boldred}[1]{#1}

\graphicspath{{imgs/}}

\newenvironment{proofsk}{\noindent\textsl{Proof (sketch).\ }}{\qed}

\title{A Declarative Framework for Specifying and Enforcing Purpose-aware Policies}

\titlerunning{Purpose-aware monitoring of \\
   authorization constraints}

\author{
Riccardo De Masellis\inst{1} \and Chiara Ghidini\inst{2} \and
  Silvio Ranise\inst{2}
}
\authorrunning{R. De Masellis, C. Ghidini,  S. Ranise}

\institute{
Trento RISE,
Via Sommarive 18, 38123 Trento, Italy\\
\email{r.demasellis@trentorise.eu}\\
\and
Bruno Kessler Foundation, 
Via Sommarive 18, 38123 Trento, Italy\\
\email{(ghidini|ranise)@fbk.eu}
}

\begin{document}

\maketitle
\sloppy

\begin{abstract}
 %!TEX root = ./main.tex

Purpose is crucial for privacy protection as it makes users confident
that their personal data are processed as intended. Available
proposals for the specification and enforcement of purpose-aware
policies are unsatisfactory for their ambiguous semantics of purposes
and/or lack of support to the run-time enforcement of policies. 

In this paper, we propose a declarative framework based on a
first-order temporal logic that allows us to give a precise semantics
to purpose-aware policies and to reuse algorithms for the design of a
run-time monitor enforcing purpose-aware policies. 
We also show the complexity of the generation and use of the monitor
which, to the best of our knowledge, is the first such a result in
literature on purpose-aware policies.

\end{abstract}

\begin{sloppypar}
%!TEX root = ./main.tex

%% \todo[inline]{ESORICS 16 pages excluding biblio and appendices, 20 in
%%   total, deadline April 4th, (11:59 PM American Samoa time).}
%% \todo[inline]{ISC 18 pages including biblio and appendices,
%%   deadline April 27.}

\section{Introduction}

An important aspect of privacy protection is the specification and
enforcement of purposes, i.e.\ users should be confident that their
data are processed as intended.  For instance, email addresses are
used only for billing but not for marketing purposes.  Unfortunately,
as already observed several times in the literature (see,
e.g.,~\cite{JafariEtal:ACM2014} for a thorough discussion), both
specifying and enforcing purposes turn out to be difficult tasks.

%%%%%%%%%%%%%%%%%%%%%%%%%%%%%%%%%%%%%%%%%%%%%%%%%%%%%%%%%%%%%%%%%%%%%%%%
\noindent {\textbf{Specification}}.
%%%%%%%%%%%%%%%%%%%%%%%%%%%%%%%%%%%%%%%%%%%%%%%%%%%%%%%%%%%%%%%%%%%%%%%%
Following the seminal paper~\cite{westin}, the specification of
privacy constraints consists of establishing when, how, and to what
extent information about people is communicated to others.  In the
context of IT systems, this amounts to define policies governing the
release of personal data for a given purpose.  Such policies are
contributed by the users of a system, also called data owners, and the
organization running the system.  The latter, besides enforcing their
own policies, are also required to take into account laws,
regulations, and directives imposed by local, national, and community
governamental bodies.  From a technological point of view, these
policies are usually mapped to access control policies augmented with
purpose constraints, which we call \emph{purpose-aware policies}
(sometimes called privacy-aware access control policies in the
literature, see, e.g.,~\cite{ardagna-etal}).  In this paper, we do not
consider the problem of deriving purpose-aware policies from the
high-level and heterogenous privacy requirements mentioned above.  We
assume that this has been done and focus instead on the basic building
blocks of the models and specification languages underlying the
policies.
As observed in~\cite{wpes09,JafariEtal:ACM2014}, such building blocks
are data-centric and rule-centric policies.  In the former, every
piece of information is associated with the purposes for which it can
be used; examples are the policies in~\cite{byun-sacmat2005} or those
expressed in the XACML Privacy
Profile.\footnote{\texttt{\url{docs.oasis-open.org/xacml/3.0/xacml-3.0-privacy-v1-spec-cd-03-en.pdf}}}
In the latter, rules specify under which conditions subjects can
perform some action on a given piece of information for some purpose;
examples are those in~\cite{purbac,ardagna-etal} and those expressed
in
EPAL.\footnote{\texttt{\url{www.w3.org/Submission/2003/SUBM-EPAL-20031110}}}
Data-centric policies more easly support the expression of the privacy
preferences of data owners while rule-centric policies permit the
expression of complex constraints
%%%%%%%%%%%%%%%%%%%%%%%%%%%%%%%%%%%%%%%%%%%%%%%%%%%%%%%%%%%%%%%%%%%%%%%%%%%%%%
%% derived from legislations or regulations that are
%%%%%%%%%%%%%%%%%%%%%%%%%%%%%%%%%%%%%%%%%%%%%%%%%%%%%%%%%%%%%%%%%%%%%%%%%%%%%%
imposed by laws, regulations, and best-practices adopted by
organizations to handle personal data.  Thus, both data- and
rule-centric policies should be supported for the specification of
purpose-aware policies.

One of the most serious problems in purpose-aware policies is the lack
of semantics for purposes, which are usually considered as atomic
identifiers. This gives rise to arbitrariness in the interpretation
of purposes; e.g., if the policy of a company states that emails of
users are collected for the purpose of communication, this allows the
organization to use emails for both billing and marketing when the
majority of users has a strong preference for the first interpretation
only. To solve this problem, several works have observed that
``\emph{an action is for a purpose if it is part of a plan for
  achieving that purpose}''~\cite{tschantz-etal}.  Among the many
possible ways to describe plans, one of the most popular is to use
workflows, i.e.\ collections of activities (called tasks) together
with their causal relationships, so that the successful termination of
a workflow corresponds to achieving the purpose which it is associated
to.
We embrace this interpretation of purpose and avoid ambiguities in its
specification by using a temporal logic which allows us to easily
express the causal relationships among actions in workflows associated
to purposes.  Additionally, the use of a logic-based framework allows
us to express in a uniform way, besides purpose specifications, also
authorization (namely, data- and rule-centric) policies together with
authorization constraints, such as Separation/Bound of Duties
(\sod/\bod).  While temporal logics have been used before for the
specification of authorization policies (see, e.g.,~\cite{ltl-rbac})
and of workflows (see, e.g.,~\cite{sttt-crampton}), it is the first
time---to the best of our knowledge---that this is done for both in
the context of purpose-aware policies.  In particular, the capability
of specifying \sod or bod constraints---which are crucial to capture
company best practices and legal requirements---seems to be left as
future work in the comprehensive framework recently proposed
in~\cite{JafariEtal:ACM2014}.
%%

%%%%%%%%%%%%%%%%%%%%%%%%%%%%%%%%%%%%%%%%%%%%%%%%%%%%%%%%%%%%%%%%%%%%%%%%%%%%%%
\noindent {\textbf{Enforcement}}.
%%%%%%%%%%%%%%%%%%%%%%%%%%%%%%%%%%%%%%%%%%%%%%%%%%%%%%%%%%%%%%%%%%%%%%%%%%%%%%
Enforcing purpose-aware policies amounts to check if (\textbf{C1}) a
user can peform an action on a certain data for a given purpose and
(\textbf{C2}) the purpose for which a user has accessed the data can
be achieved.

(\textbf{C1}) is relatively easy and well-understood being an
extension of mechanisms for the enforcement of access control policies
(see, e.g.,~\cite{samarati-vimercati} for an overview) by considering
%%%%%%%%%%%%%%%%%%%%%%%%%%%%%%%%%%%%%%%%%%%%%%%%%%%%%%%%%%%%%%%%%%%%%%%%%%%%%%
%% whereby both .  In our framework, this activity is slightly more
%% complex because of
%%%%%%%%%%%%%%%%%%%%%%%%%%%%%%%%%%%%%%%%%%%%%%%%%%%%%%%%%%%%%%%%%%%%%%%%%%%%%%
the combined effect of rule- and data-centric policies.  This is so
because tasks are executed under the responsibility of users and tasks
may perform actions on data.  To permit the execution of a task $t$ by
a user $u$, it is thus necessary to check if $u$ can perform all the
actions $\mathit{Acts}$ associated to $t$ (under a given rule-centric
policy) and the data---on which $\mathit{Acts}$ are performed---are
released by their data owners (under a given data-centric policy).

(\textbf{C2}) is much more complex than (\textbf{C1}) as it requires
to foresee if there exists an assignment of users to tasks that allows
for the successful termination of the workflow.  This is so
because---as discussed above---a purpose is associated to a workflow
so that its successful execution implies the achievement of the
purpose.  % Workflows are augmented with the so-called authorization
% constraints, such as Separation/Bound of Duty (\sod/\bod) constraints
% requiring two tasks to be executed by distinct users or the same user,
% respectively.  It is well-known that authorization constraints help
% preventing frauds (in case of \sod) or implementing company best
% practices (in case of \bod) and are thus important also for the
% specification of purpose-aware policies.
% \boldred{Workflows are also augmented with \sod and \bod constraints, which,
% being well-known of help in preventing frauds (in case of \sod) or in
% implementing company best practices (in case of \bod), are
% important for the specification of purpose-aware policies.}
The
problem of checking (offline) if a workflow can successfully
terminate, known as the Workflow Satisfiability Problem (\wsp), is
already computationally expensive with one \sod~\cite{wang-li}, and
moreover, the on-line monitoring of authorization constraints requires
to solve several instances of the \wsp~\cite{crampton}.  For
purpose-aware policies, this implies that it is necessary to solve an
instance of the \wsp per user request of executing a task in the
workflow associated to a given purpose.

%%%%%%%%%%%%%%%%%%%%%%%%%%%%%%%%%%%%%%%%%%%%%%%%%%%%%%%%%%%%%%%%%%%%%%%%
\noindent {\textbf{Contributions}}.
%%%%%%%%%%%%%%%%%%%%%%%%%%%%%%%%%%%%%%%%%%%%%%%%%%%%%%%%%%%%%%%%%%%%%%%%
The paper provides the following contributions:
\begin{compactitem}
\item The \emph{specification} of a comprehensive framework for expressing
  purpose-aware policies which are a combination of data- and
  rule-centric policies together with workflows augmented with
  authorization constraints (Section~\ref{sec:framework} and, in
  particular, Figure~\ref{fig:workflow}). To the best of our
  knowledge, this is the first time authorization constraints are
  considered and naturally integrated in a purpose-aware
  setting.
\item The \emph{semantic formalization} of purpose-aware policies
  as formulas in first-order temporal logic
  (Section~\ref{subsec:semantics}).
\item The provision of \emph{formal techniques} not only for the
  (on-line) enforcement of purpose-aware policies, but also for their
  (off-line) analysis, together with decidability and
  complexity results (Section~\ref{sec:verification}).
\end{compactitem}
% The paper describes a declarative framework to specify and enforce
% purpose-aware policies where:
% \begin{itemize}
% \item purpose-aware policies are a combination of data- and
%   rule-centric policies together with workflows augmented with
%   authorization constraints (Section~\ref{sec:framework} and, in
%   particular, Figure~\ref{fig:workflow}),
% \item the semantics of purpose-aware policies is given as formulae in
%   a first-order temporal logic (Section~\ref{subsec:semantics}),
% \item the decidability and complexity of off- and on-line verification
%   problems of purpose-aware policies are derived from those of
%   satisfiability problems in the adopted first-order temporal logic
%   (Section~\ref{sec:verification}).
% \end{itemize}
%%
The choice of a first-order (linear-time) temporal logic for the
semantics of purpose-aware policies has three main motivations.
First, since both data- and rule-centric policies can be seen as
access control policies, the use of first-order logic formulae to
express them is adequate following the established tradition of
expressing a wide range of access control policy idioms in (fragments
of) first-order logic; see,
e.g.,~\cite{constraintdatalog,arkoudas-etal}.  In the rest of this
paper, we do not elaborate further on this issue as it is well-studied
and focus on aspects more closely related to purposes.  We just remind
that answering a request of performing an action on an object by a
user can be efficiently done (linear time in the size of the
authorization query) as shown in,
e.g.,~\cite{constraintdatalog,arkoudas-etal}.  The second motivation
for choosing a temporal logic comes from the declarative approach to
the specification of workflows proposed
in~\cite{WestergaardM11,AalstPS09} which provides a formal semantics
to their executions.  The third motivation is the adaptation and reuse
of available
results~\cite{DeMasellisS13,DeGiacomoMGMM14,DeGiacomoDMM14} concerning
logical problems in a combination of first-order and linear-time
temporal logic to which verification problems for purpose-aware
policies can be translated.  Additionally, it allows us to reuse
available algorithms for solving logical problems on top of which a
module for the the enforcement of purpose-aware policies can be built.

A running example (Section~\ref{sec:motivations}) introducing the main
issues related to purpose-aware policies is used throught the paper to
illustrate the main concepts of our framework.  Related work and
conclusions are also discussed (Section~\ref{sec:related}).

  \section{Running example}
\label{sec:motivations}

We describe 
%%%%%%%%%%%%%%%%%%%%%%%%%%%%%%%%%%%%%%%%%%%%%%%%%%%%%%%%%%%%%%%%%%%%%%%%%%%
%% present a high-level overview of our framework with the help of 
%%%%%%%%%%%%%%%%%%%%%%%%%%%%%%%%%%%%%%%%%%%%%%%%%%%%%%%%%%%%%%%%%%%%%%%%%%%
a running example, based on \emph{Smart
  campus}\footnote{\url{http://www.smartcampuslab.it}}, which will be
used throughout the paper.
%%%%%%%%%%%%%%%%%%%%%%%%%%%%%%%%%%%%%%%%%%%%%%%%%%%%%%%%%%%%%%%%%%%%%%%%%%%
% All concepts and technicalities that follow will be explained in
% detail in the next Sections.
%% \begin{example}
%   Smart campus is a unique endpoint of a service-oriented
%   architecture aimed at managing transportation, events, social
%   services and education in a smart city.
%%%%%%%%%%%%%%%%%%%%%%%%%%%%%%%%%%%%%%%%%%%%%%%%%%%%%%%%%%%%%%%%%%%%%%%%%%%
Smart campus is a platform in which 
%%%%%%%%%%%%%%%%%%%%%%%%%%%%%%%%%%%%%%%%%%%%%%%%%%%%%%%%%%%%%%%%%%%%%%%%%%%
%% unique endpoint where 
%%%%%%%%%%%%%%%%%%%%%%%%%%%%%%%%%%%%%%%%%%%%%%%%%%%%%%%%%%%%%%%%%%%%%%%%%%%
citizens, institutions and companies can communicate with each other
by exchanging data and services.  It 
%%%%%%%%%%%%%%%%%%%%%%%%%%%%%%%%%%%%%%%%%%%%%%%%%%%%%%%%%%%%%%%%%%%%%%%%%%%
%% The platform 
%%%%%%%%%%%%%%%%%%%%%%%%%%%%%%%%%%%%%%%%%%%%%%%%%%%%%%%%%%%%%%%%%%%%%%%%%%%
provides functionalities to access 
%%%%%%%%%%%%%%%%%%%%%%%%%%%%%%%%%%%%%%%%%%%%%%%%%%%%%%%%%%%%%%%%%%%%%%%%%%%
%% any kind of 
%%%%%%%%%%%%%%%%%%%%%%%%%%%%%%%%%%%%%%%%%%%%%%%%%%%%%%%%%%%%%%%%%%%%%%%%%%%
information about
%%%%%%%%%%%%%%%%%%%%%%%%%%%%%%%%%%%%%%%%%%%%%%%%%%%%%%%%%%%%%%%%%%%%%%%%%%%
 %% , such as
%%%%%%%%%%%%%%%%%%%%%%%%%%%%%%%%%%%%%%%%%%%%%%%%%%%%%%%%%%%%%%%%%%%%%%%%%%%
transportations, social services, education, and user %% social networks
profiles.
%%%%%%%%%%%%%%%%%%%%%%%%%%%%%%%%%%%%%%%%%%%%%%%%%%%%%%%%%%%%%%%%%%%%%%%%%%%
%% It also provides modular services in a service-oriented
%% (\textsc{soa}) fashion.
%%%%%%%%%%%%%%%%%%%%%%%%%%%%%%%%%%%%%%%%%%%%%%%%%%%%%%%%%%%%%%%%%%%%%%%%%%%
These services allow companies to build new applications and thus
offer new services to citizens. In this kind of scenario, personal
data of users is the main ingredient to provide customized services.
Indeed, users should be confident that the services access only the
data required to their needs and use them for the right purposes.
Additionally, service providers should comply with laws, regulations,
and best practices in handling data mandated by local governments, the
European Union, and enterprises.  In other words, access to personal
data must be mediated by appropriate authorization policies augmented
with purpose constraints so that only authorized subjects have the right
to access certain data for a given purpose.  Following an established
line of works (see, e.g.,~\cite{JafariEtal:ACM2014} for an overview),
we assume that the purpose of an action is determined by its
relationships with other, interrelated, actions.  

For concreteness, we illustrate these ideas by
%%%%%%%%%%%%%%%%%%%%%%%%%%%%%%%%%%%%%%%%%%%%%%%%%%%%%%%%%%%%%%%%%%%%%%%%%%%
%% the need for protecting data from unauthorized access or usage is of
%% primary concern.
%%%%%%%%%%%%%%%%%%%%%%%%%%%%%%%%%%%%%%%%%%%%%%%%%%%%%%%%%%%%%%%%%%%%%%%%%%%
%% We illustrate this kind of situations by 
%%%%%%%%%%%%%%%%%%%%%%%%%%%%%%%%%%%%%%%%%%%%%%%%%%%%%%%%%%%%%%%%%%%%%%%%%%%
considering the situation in which some personal data of users in the
Smart campus platform (namely, the work experience and the academic
transcripts) is accessed by JH, a job hunting company, for the purpose
of finding jobs to students.
%%%%%%%%%%%%%%%%%%%%%%%%%%%%%%%%%%%%%%%%%%%%%%%%%%%%%%%%%%%%%%%%%%%%%%%%%%%
\begin{figure}[t]
\centering
    \includegraphics[width=\textwidth]{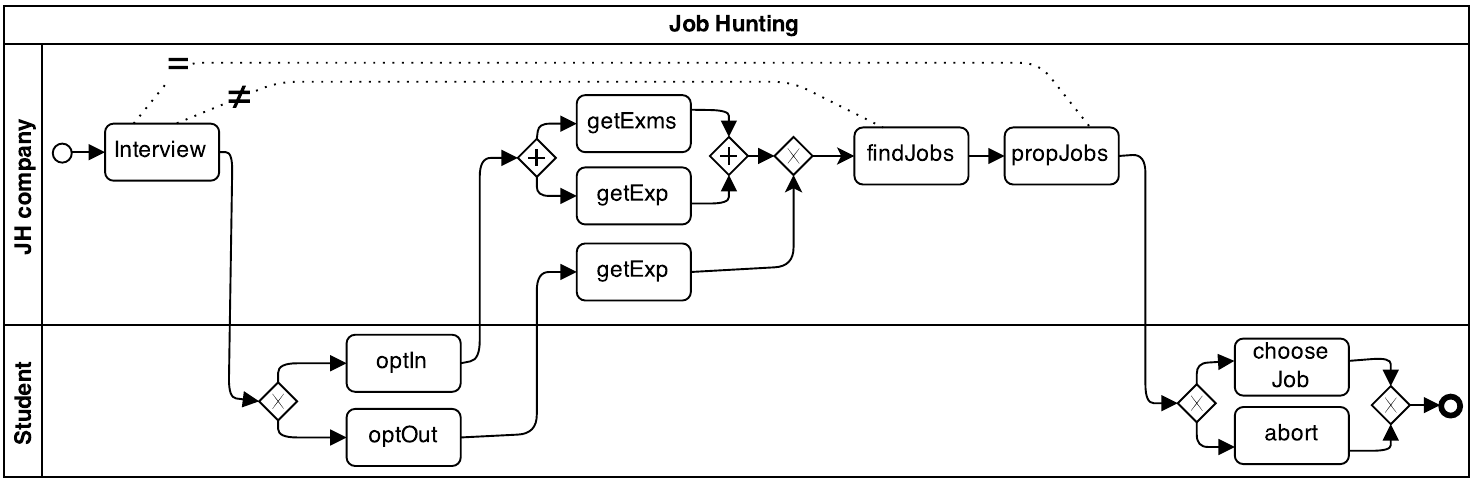}
\begin{scriptsize}
\begin{tabular}{ | C{1.5cm} | C{2.5cm} | C{2.7cm} | C{3cm} | C{1.8cm} | }
\hline
  $\msf{interview}$
  &
$\begin{array}{c}
\Diamond \msf{optIn} \lor \Diamond \msf{optOut}
\\ \land \\ 
\neg (\Diamond \msf{optIn} \land \Diamond \msf{optOut})
\end{array}$ 
&
$\begin{array}{c} 
(\neg(\msf{getExms} \lor \msf{getExp}) \\ \Until \\
(\msf{optIn} \lor \msf{optOut})) \\ \land \\
(\neg \msf{getExsm} \Until \msf{optIn})
\end{array}$
&
$\begin{array}{c}
(\neg \msf{findJobs} \\ \Until \\
(\msf{getExms} \lor
  \msf{getExp})) \\ \land \\
(\Diamond \msf{findJobs}
\\ \land \\
\msf{findJobs} \ra \Next \msf{propJobs})
\end{array}$
& 
$\begin{array}{c}
\Diamond (\msf{chooseJob} \\ \lor \\ \msf{abort})
\end{array}$ \\ \hline
\end{tabular}
\end{scriptsize}
\caption{\label{fig:workflow}The $\msf{JobHunting}$ workflow expressed
  in \bpmn (upper half) and as a (partial) set of \LTL
  formulae (lower half).}
\end{figure}
%%%%%%%%%%%%%%%%%%%%%%%%%%%%%%%%%%%%%%%%%%%%%%%%%%%%%%%%%%%%%%%%%%%%%%%%%%%
% One of the smart campus objectives in the education and work field
% concerns bridging the gap between universities and the labour market.
%%%%%%%%%%%%%%%%%%%%%%%%%%%%%%%%%%%%%%%%%%%%%%%%%%%%%%%%%%%%%%%%%%%%%%%%%%%
JH has deployed in the Smart Campus platform the service depicted of
Figure~\ref{fig:workflow} (upper half),
%%%%%%%%%%%%%%%%%%%%%%%%%%%%%%%%%%%%%%%%%%%%%%%%%%%%%%%%%%%%%%%%%%%%%%%%%%%
% is part of the community and  owns the
% necessary expertise to find jobs that fit clients preferences.
%  The idea is that JH can join smart
% campus in order to use the platform functionalities and be contacted
% by students also being part of the community.
% As the platform provides
% \textsc{api} to access students' data, it has to guarantee that
% registered companies process data in compliance with the clients
% policies. When JH requests to join smart campus as a service provider,
% its processes are (possibly manually) analyzed to check whether the
% usage of data meets the objectives of the process: intuitively, it
% would make no sense to access medical records for the process of job
% hunting.
%%
%% delivers the $\msf{JobHunting}$ service, which workflow is pictured in
%% Figure~\ref{fig:workflow} on top, 
%%%%%%%%%%%%%%%%%%%%%%%%%%%%%%%%%%%%%%%%%%%%%%%%%%%%%%%%%%%%%%%%%%%%%%%%%%%
specified in the Business Process Model and Notation (\bpmn).
The swim lane labelled `Student' contains the activities (also called
tasks) that must be executed under the responsibility of the data
owner and the swim lane labelled `JH company' shows the activities that
employees of JH are supposed to perform for the purpose of finding
some jobs to the data owner.
%%%%%%%%%%%%%%%%%%%%%%%%%%%%%%%%%%%%%%%%%%%%%%%%%%%%%%%%%%%%%%%%%%%%%%%%%%%
% Tasks in different swim lanes are executed by different actors. 
%%%%%%%%%%%%%%%%%%%%%%%%%%%%%%%%%%%%%%%%%%%%%%%%%%%%%%%%%%%%%%%%%%%%%%%%%%%
First of all, an employee of JH performs an $\msf{interview}$ to the
student to understand his/her job preferences.  Then, the student
decides to give or not his/her consent for JH to access his/her
academic transcripts by executing either task $\msf{optIn}$ or task
$\msf{optOut}$, respectively (the empty diamond before the two tasks
in Figure~\ref{fig:workflow} is an exclusive-or gateway).
%%%%%%%%%%%%%%%%%%%%%%%%%%%%%%%%%%%%%%%%%%%%%%%%%%%%%%%%%%%%%%%%%%%%%%%%%%%
%% and then asking him whether give to JH the consent ($\msf{optIn}$
%% or $\msf{optOut}$) to access his university data.
%%%%%%%%%%%%%%%%%%%%%%%%%%%%%%%%%%%%%%%%%%%%%%%%%%%%%%%%%%%%%%%%%%%%%%%%%%%
If he/she opts in, an employee of JH can access both his/her academic
transcripts and past experience (by executing both $\msf{getExms}$ and
$\msf{getExp}$ since the diamond containing the plus sign before the
two tasks in Figure~\ref{fig:workflow} is a parallel gateway);
otherwise, only the student past experiences can be accessed.  Based
on the interview and the collected personal information, an employee
of JH search for jobs the student can be interested in (task
$\msf{findJobs}$) and some other employee proposes him/her some of
them (task $\msf{propJobs}$). Finally, the students decides if
choosing one of the jobs or to abort the process (tasks
$\msf{chooseJob}$ and $\msf{abort}$, respectively).  Further
authorization constraints are imposed on which employees can execute
task $\msf{interview}$, $\msf{findJobs}$, and $\msf{proposeJobs}$: the
first two must be executed by different employees to keep the overall
process unbiased---this is called a Separation of Duty (\sod)
constraint---whereas the first and last tasks must be executed by the
same employee so that the student gets in contact with the same person
of JH---this is called a Bind of Duty (\bod) constraint.  (In
Figure~\ref{fig:workflow}, these constraints are shown as dotted lines
connecting the tasks labelled by the distinct $\neq$ or equal $=$ sign
in case of \sod or \bod, respectively.)
%%%%%%%%%%%%%%%%%%%%%%%%%%%%%%%%%%%%%%%%%%%%%%%%%%%%%%%%%%%%%%%%%%%%%%%%%%%
%% Furthermore, we include \sod and \bod constraints, which forces two
%% (or more) tasks to be executed by different, or the same, subject
%% respectively. In the $\msf{JobHunting}$ workflow, to keep the overall
%% process unbiased, the subject performing the interview cannot
%% participate in the search for jobs.
%% %%
%% Similarly, in order to
%% establish a more personal interactions with clients, JH requires that
%% the employee performing the interview must be the same which proposes
%% possible job opportunities to the student.
%%%%%%%%%%%%%%%%%%%%%%%%%%%%%%%%%%%%%%%%%%%%%%%%%%%%%%%%%%%%%%%%%%%%%%%%%%%

To summarize, the workflow specification is used to specify the
purpose of an activity (task) with respect to all the others that must
be executed for achieving the given purpose.
From now on, we assume that a workflow is uniquely associated to a
purpose or, equivalently, that the semantics of a purpose is its
associated workflow.  

Since the tasks in the workflow are executed under the responsibility
of a user (e.g., an employee of JH), he/she must have the right to
access such data.  For instance, the task $\msf{getExp}$ takes as
input the list of past job experiences of the student.  The employee
of JH executing this task must have the right to access such a list
and the student should have given the consent to access this
information to (an employee of) JH.  In other words, every time an
employee of JH asks to execute an activity for the purpose of finding
jobs, he/she not only must have the right to do so according to the
authorization policy of the company but also the student (data owner)
should agree to release the information for the purpose of
finding jobs.  In other words, there are two types of policies that
must be taken into account when granting the right to execute a task
to an employee: one is called \emph{rule-centric}, and constrains
access by considering subjects, actions, and data objects while the
other is called \emph{data-centric}, and is such that data owners
constrain access to their data objects for certain purposes only.  
%%%%%%%%%%%%%%%%%%%%%%%%%%%%%%%%%%%%%%%%%%%%%%%%%%%%%%%%%%%%%%%%%%%%%%%%%%%%
%% Since data- and rule-centric policies are expressed over data, we provide 
%%
%%
%%
% In order to specify such policies, a description and structure of data
% that will be accessed by the processes is needed. 
% The following
%% the data objects of our frameworks, which are the relations of
%% the university database:
%% \begin{compactitem}
%% \item $UserProfile$,
%% %%\item $AddressBook$,
%% \item $UniversityData$,
%%   %% \item $Exams(ssn, grade)$,
%% \item $JobExperiences$
%% \end{compactitem}
%% which store information on students' personal data, academic careers
%% and job experiences.
%%%%%%%%%%%%%%%%%%%%%%%%%%%%%%%%%%%%%%%%%%%%%%%%%%%%%%%%%%%%%%%%%%%%%%%%%%%%
% Rule- and data-centric policies are specified as tuples:
% \[
% \begin{array}{c}
%   rcp(\msf{sam}, \msf{write}, \msf{JobExperiences}), \\
%   dcp(\msf{UniversityData}, \msf{sam}, \msf{jobHunting})
% \end{array}
% \]
%%%%%%%%%%%%%%%%%%%%%%%%%%%%%%%%%%%%%%%%%%%%%%%%%%%%%%%%%%%%%%%%%%%%%%%%%%%%
For instance,
%%%%%%%%%%%%%%%%%%%%%%%%%%%%%%%%%%%%%%%%%%%%%%%%%%%%%%%%%%%%%%%%%%%%%%%%%%%%
%% Finally, we specify rule- and data-centric policies. An example of the
%%%%%%%%%%%%%%%%%%%%%%%%%%%%%%%%%%%%%%%%%%%%%%%%%%%%%%%%%%%%%%%%%%%%%%%%%%%%
in the job hunting scenario, the rule-centric policy specifies that
%%%%%%%%%%%%%%%%%%%%%%%%%%%%%%%%%%%%%%%%%%%%%%%%%%%%%%%%%%%%%%%%%%%%%%%%%%%%
%% former is that 
%%%%%%%%%%%%%%%%%%%%%%%%%%%%%%%%%%%%%%%%%%%%%%%%%%%%%%%%%%%%%%%%%%%%%%%%%%%%
employee $\msf{bob}$ has the right to $\msf{read}$ the list of job
experiencens of students and
%%%%%%%%%%%%%%%%%%%%%%%%%%%%%%%%%%%%%%%%%%%%%%%%%%%%%%%%%%%%%%%%%%%%%%%%%%%%
%% on relation $\msf{JobExperiences}$, while one of the
%% latter says
%%%%%%%%%%%%%%%%%%%%%%%%%%%%%%%%%%%%%%%%%%%%%%%%%%%%%%%%%%%%%%%%%%%%%%%%%%%%
the data-centric policy specifies that student $\msf{sam}$'s academic
transcripts
%%%%%%%%%%%%%%%%%%%%%%%%%%%%%%%%%%%%%%%%%%%%%%%%%%%%%%%%%%%%%%%%%%%%%%%%%%%%
%% $\msf{UniversityData}$
%%%%%%%%%%%%%%%%%%%%%%%%%%%%%%%%%%%%%%%%%%%%%%%%%%%%%%%%%%%%%%%%%%%%%%%%%%%%
can be accessed for the purpose of $\msf{JobHunting}$.
%%%%%%%%%%%%%%%%%%%%%%%%%%%%%%%%%%%%%%%%%%%%%%%%%%%%%%%%%%%%%%%%%%%%%%%%%%%%
%%
% \[uses(\msf{getExams}, \msf{read}, \msf{UniveristyData})\]
%%
%% \end{example}
%%%%%%%%%%%%%%%%%%%%%%%%%%%%%%%%%%%%%%%%%%%%%%%%%%%%%%%%%%%%%%%%%%%%%%%%%%%%

%%%%%%%%%%%%%%%%%%%%%%%%
Notice the subtle interplay between purposes, described by workflows,
and authorization policies.  
%%%%%%%%%%%%%%%%%%%%%%%%%%%%%%%%%%%%%%%%%%%%%%%%%%%%%%%%%%%%%%%%%%%%%%%%%%%%
%% We observe that workflows, apart from specifying purposes, can be
%% notably used to refine data-centric and rule-centric
%% policies. Indeed,
%% While policies specify contraints that must be satisfied by {every
%%   execution} of a workflow, in our framework different privacy 
%%%%%%%%%%%%%%%%%%%%%%%%%%%%%%%%%%%%%%%%%%%%%%%%%%%%%%%%%%%%%%%%%%%%%%%%%%%%
For instance, the execution of certain tasks can modify both rule- and
data-centric policies as it is the case of $\msf{optIn}$ and
$\msf{optOut}$ in Figure~\ref{fig:workflow}.  
%%%%%%%%%%%%%%%%%%%%%%%%%%%%%%%%%%%%%%%%%%%%%%%%%%%%%%%%%%%%%%%%%%%%%%%%%%%%
%% policies \emph{for each workflow execution} can be defined.
%%%%%%%%%%%%%%%%%%%%%%%%%%%%%%%%%%%%%%%%%%%%%%%%%%%%%%%%%%%%%%%%%%%%%%%%%%%%
When executing the latter, the execution of the task $\msf{getExms}$
is skipped in the current instance of the workflow despite the fact
that $\msf{sam}$ has agreed to disclose such information according to
the above data-centric policy (for which $\msf{sam}$'s academic
transcripts can always be used for the purpose of $\msf{jobHunting}$).
This flexibility allows us to model certain data directive for privacy
(see, e.g.,~\cite{EU-Directive95/46/EC}) in which the data owner must
explicitly give his/her consent to access his/her personal data, every
time it is requested.
%%%%%%%%%%%%%%%%%%%%%%%%%%%%%%%%%%%%%%%%%%%%%%%%%%%%%%%%%%%%%%%%%%%%%%%%%%%%
%% coarse-grained policy: indeed, by choosing $\msf{optOut}$ which forces
%% $\msf{getExms}$ to be skipped, $\msf{sam}$ can decide to hide his
%% university data within the \emph{current} workflow execution.
%%%%%%%%%%%%%%%%%%%%%%%%%%%%%%%%%%%%%%%%%%%%%%%%%%%%%%%%%%%%%%%%%%%%%%%%%%%%

Finally, observe that handling purposes in presence of authorization
constraints (such as \sod or \bod) requires to solve, at run-time, the
\emph{Workflow Satisfiability Problem} (\wsp)~\cite{crampton}, i.e., to
be able to answer the question: does there exist an assignment of
authorized users (according to the rule- and data-centric policies) to
workflow tasks that satisfies the authorization constraints (\sod or
\bod)?  The \wsp is known to be a computationally expensive activity;
it is already NP-hard with one \sod constraint~\cite{wang-li}.  To
make things even more complex, at run-time, we need to solve several
instances of the \wsp, one per user request of executing a task in a
workflow associated to a purpose.  This is so because to solve an
instance of the \wsp, it is sufficient to find an execution sequence of
task-user pairs containing all the tasks in a workflow.  In the job
hunting scenario, consider the initial request of an employee $u_0$ to
execute task $\msf{interview}$.  This implies to check whether there
exists a sequence $\pi=(\msf{interview}, u_0), \pi'$ of
task-user pairs form which is a successfull execution of the workflow.
For instance, if $\pi'$ is $(\msf{optOut},\msf{sam}),
(\msf{getExp},u_0), (\msf{findJobs},u_1), (\msf{propJobs},u_0),
(\msf{abort},\msf{sam})$, then $\pi$ represents a succesfull
execution of the workflow when $\msf{sam}$ decides to opt out and not
releasing his academic transcript to JH.  However, this implies only
that the \wsp is solvable and the sequence represents a possible future
execution.  Imagine now that $\msf{sam}$ decides to opt in: $\pi'$
is useless and a new instance of the \wsp must be solved to find a
successful execution of the workflow of the form $(\msf{interview},
u_0),(\msf{optIn},\msf{sam}),\pi''$ for some sequence $\pi''$.
Notice that, even if $\msf{sam}$ decides to opt out, the sequence
$\pi'$ may be useless, provided that a user $u_2$ distinct from
$u_0$ is granted the possibility to execute task $\msf{getExp}$.

\section{A Declarative Framework for Purpose-aware Policies}
\label{sec:framework}
%%%%%%%%%%%%%%%%%%%%%%%%%%%%%%%%%%%%%%%%%%%%%%%%%%%%%%%%%%%%%%%%%%%%%%%%%%%%%%
%% We now enter into the technical details of our framework by
%% formalizing rule-, data- and purpose-centric policies.
%%%%%%%%%%%%%%%%%%%%%%%%%%%%%%%%%%%%%%%%%%%%%%%%%%%%%%%%%%%%%%%%%%%%%%%%%%%%%%
%%%%%%%%%%%%%%%%%%%%%%%%%%%%%%%%%%%%%%%%%%%%%%%%%%%%%%%%%%%%%%%%%%%%%%%%%%%%%%
\begin{figure}[t]
  \centering
  \begin{tabular}{cc}
  \includegraphics[scale=0.6]{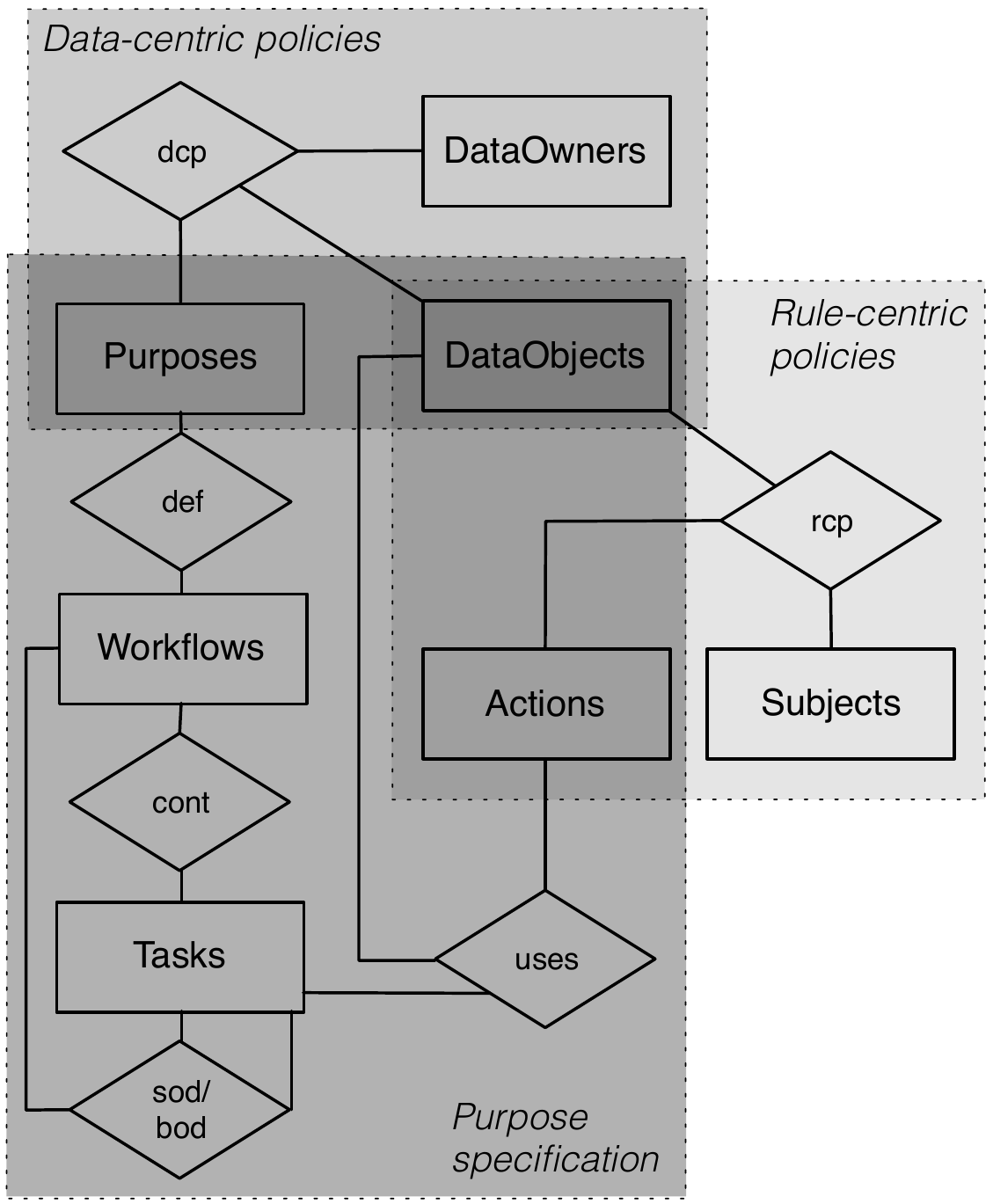} &
  \raisebox{4.25cm}{
  \hspace{-2cm}
  \begin{minipage}{.6\textwidth}
  \begin{eqnarray*}
  \mathit{dcp} &\subseteq& \mathit{DataObjects} \times 
                           \mathit{DataOwners}  \times 
                           \mathit{Purposes} \\
  \mathit{rcp} &\subseteq& \mathit{Subjects}   \times 
                           \mathit{Actions}    \times 
                           \mathit{DataObjects} \\
  && \\
  && \\
  && \\
  && \\
  && \\
  && \\
  && \\
  && \\
  && \\
  &&\\
\mathit{def}&:& \mathit{Purposes} \ra \mathit{Workflows} \\
  \mathit{cont} &\subseteq& \mathit{Workflow} \times 
                           \mathit{Tasks} \\
  \mathit{uses} &\subseteq& \mathit{Tasks} \times 
                            \mathit{Actions} \times 
                            \mathit{DataObjects} \\
  \mathit{sod} &\subseteq& \mathit{Workflows} \times 
                         \mathit{Tasks} \times 
                         \mathit{Tasks} \\
  \mathit{bod} &\subseteq& \mathit{Workflows} \times 
                         \mathit{Tasks} \times 
                         \mathit{Tasks}
  \end{eqnarray*}
  \end{minipage}}
  \end{tabular}
  \caption{\label{fig:ER}Conceptual representation of the data-, rule- and
    purpose-centric policies.}
\end{figure}
%%%%%%%%%%%%%%%%%%%%%%%%%%%%%%%%%%%%%%%%%%%%%%%%%%%%%%%%%%%%%%%%%%%%%%%%%%%%%%
On the left of Figure~\ref{fig:ER}, it is shown an entity-relationship
diagram describing the conceptual organization of our framework
(rectangles represent sets of entities and diamonds relationships
among them).  We have $\mathit{DataOwners}$ who own
$\mathit{DataObjects}$ and decide the $\mathit{Purposes}$ for which
these can be accessed by means of a data-centric policy (relation
$\mathit{dcp}$).  $\mathit{Subjects}$ can perform certain
$\mathit{Actions}$ on $\mathit{DataObjects}$ according to a
rule-centric policy (relation $\mathit{rcp}$).  $\mathit{Purposes}$
are defined (relation $\mathit{def}$) in terms of $\mathit{Workflows}$
which are composed of $\mathit{Tasks}$ (relation $\mathit{cont}$) and
\sod or \bod constraints; each task can perform some
$\mathit{Actions}$ on $\mathit{DataObjects}$ (relation
$\mathit{uses}$).

On the right of Figure~\ref{fig:ER}, it is shown the formal
characterization of the relationships as subsets of the cartesian
products of the appropriate sets of entities.  To illustrate, recall
the running example in Section~\ref{sec:motivations}:
\begin{itemize}
\item the rule- and data-centric policies ``$\msf{bob}$ has the right to
  read the list of job experiencens of students'' and ``$\msf{sam}$'s
  academic transcripts can be accessed for the purpose of
  $\msf{JobHunting}$'' can be specified by relations $\mathit{rcp}$
  and $\mathit{dcp}$ which are such that\footnote{Given an $n$-ary
    relation $R$, we write $R(e_1, ..., e_n)$ for $(e_1, ..., e_n)\in
    R$.} $\mathit{rcp}(\msf{bob}, \msf{read}, \msf{jobExpList})$ and
  $\mathit{dcp}(\msf{academicTranscript}, \msf{sam},
  \msf{jobHunting})$;
\item if $\varphi$ is the specification of the workflow in the upper
  half of Figure~\ref{fig:workflow} (we explain below what is
  $\varphi$), then $\mathit{def}(\msf{JobHunting})=\varphi$ (notice
  that $\mathit{def}$ is a total function from $\mathit{Purposes}$ to
  $\mathit{Workflows}$, i.e.\ every purpose is associated to a
  workflow);
\item the \sod and \bod contraints in Figure~\ref{fig:workflow} can be
  specified by relations $\mathit{sod}$ and $\mathit{bod}$ such that 
%%%%%%%%%%%%%%%%%%%%%%%%%%%%%%%%%%%%%%%%%%%%%%%%%%%%%%%%%%%%%%%%%%%%%%%%%%
%%  are triples $sod \subseteq Workflows \times Tasks \times Tasks$,
%% resp., $bod \subseteq Workflows \times Tasks \times Tasks$, saying
%% which couple of tasks must be executed by different, resp., the
%% same, user in a workflow. The two contraints of the
%%%%%%%%%%%%%%%%%%%%%%%%%%%%%%%%%%%%%%%%%%%%%%%%%%%%%%%%%%%%%%%%%%%%%%%%%%
  $sod(\varphi, \msf{interview}, \msf{findJobs})$ and $bod(\varphi,
  \msf{interview}, \msf{propJobs})$;
\item the fact that, for example, tasks $\msf{interview}$ and
  $\msf{optIn}$ are part of the worflow specification $\varphi$ can be
  specified by a relation $\mathit{cont}$ such that
  %%%%%%%%%%%%%%%%%%%%%%%%%%%%%%%%%%%%%%%%%%%%%%%%%%%%%%%%%%%%%%%%%%%%%%%%%%
  %% Relation $cont \subseteq Workflow \times Tasks$, instead, relates a
  %% $Workflow$ with the set of $Tasks$ it contains. An example is
  %%%%%%%%%%%%%%%%%%%%%%%%%%%%%%%%%%%%%%%%%%%%%%%%%%%%%%%%%%%%%%%%%%%%%%%%%%
  $\mathit{cont}(\varphi, \msf{interview})$ and
  $\mathit{cont}(\varphi, \msf{optIn})$;
\item the fact that the task $\msf{interview}$ reads the user profile
  can be specified by a relation $\mathit{uses}$ such that 
  %%%%%%%%%%%%%%%%%%%%%%%%%%%%%%%%%%%%%%%%%%%%%%%%%%%%%%%%%%%%%%%%%%%%%%%%%%
  %% Finally relation $uses \subseteq Tasks \times Actions \times
  %% DataObjects$ expresses which $Actions$ are used by $Tasks$ on
  %% which $DataObjects$. In such a way, we do not pose any
  %% restrictions on how tasks are structured, being interested only
  %% on which data they access. Therefore, tasks can play the role of
  %% atomic actions as well as automated complex programs, external
  %% (\textsc{soa}) services or tasks involving human activities
  %% (called human tasks in the \bpm literature). As an example,
  %%%%%%%%%%%%%%%%%%%%%%%%%%%%%%%%%%%%%%%%%%%%%%%%%%%%%%%%%%%%%%%%%%%%%%%%%%
  $\mathit{uses}(\msf{interview}, \msf{read}, \msf{UserProfile})$;
  %%%%%%%%%%%%%%%%%%%%%%%%%%%%%%%%%%%%%%%%%%%%%%%%%%%%%%%%%%%%%%%%%%%%%%%%%%
  %% means that action $\msf{read}$ on $\msf{UserProfile}$ is used by
  %% task $\msf{interview}$.
  %%%%%%%%%%%%%%%%%%%%%%%%%%%%%%%%%%%%%%%%%%%%%%%%%%%%%%%%%%%%%%%%%%%%%%%%%%
\end{itemize}
%%%%%%%%%%%%%%%%%%%%%%%%%%%%%%%%%%%%%%%%%%%%%%%%%%%%%%%%%%%%%%%%%%%%%%%%%%
%% We distinguish among $Subjects$, representing who access data;
%% $DataObjects$, the data themselves; $Actions$, the allowed
%% operations on data; $Workflows$, describing the high-level business
%% processes; $Tasks$, the basic workflow building blocks and
%% $Purposes$, the reason data are accessed for. Relations among these
%% concepts constitutes the policies.
%%
%% \paragraph{Rule-centric policies}
%% specifies which $Subjects$ can access which (data) $Objects$ with
%% which $Actions$, hence they are a relation
%% $rcp \subseteq Subjects \times Actions \times DataOwners$.  
%%%%%%%%%%%%%%%%%%%%%%%%%%%%%%%%%%%%%%%%%%%%%%%%%%%%%%%%%%%%%%%%%%%%%%%%%%
%%%%%%%%%%%%%%%%%%%%%%%%%%%%%%%%%%%%%%%%%%%%%%%%%%%%%%%%%%%%%%%%%%%%%%%%%%
%% \paragraph{Data-centric policies} define which are the $Purposes$
%% $DataObject$ of $DataOwners$ can be accessed for, so:
%% $dcp \subseteq DataObjets \times DataOwner \times Purposes$. 
%%%%%%%%%%%%%%%%%%%%%%%%%%%%%%%%%%%%%%%%%%%%%%%%%%%%%%%%%%%%%%%%%%%%%%%%%%
%%
%% \paragraph{Purpose-centric policies} are more complex and relates
%% purposes, workflows, task, actions, data objects and \sod and \bod
%% contraints. Relation $def$ is actually a \emph{total function}
%% $def: Purposes \ra Workflows$ associating $Purposes$, which are
%% basically names, with workflows, 
%%%%%%%%%%%%%%%%%%%%%%%%%%%%%%%%%%%%%%%%%%%%%%%%%%%%%%%%%%%%%%%%%%%%%%%%%%
We now explain how we specify workflows in our framework.  Following
the declarative approach in~\cite{AalstPS09}, we have chosen
Linear-time Temporal Logic (\LTL) as the specification language.  The
main reason for this is two-fold.  First, well-known techniques (see,
e.g.,~\cite{kroger-merz}) are available to translate procedural
descriptions of workflows (e.g., that in the upper half of
Figure~\ref{fig:workflow}), and more in general concurrent systems, to
\LTL formulae.  For instance, the lower half of
Figure~\ref{fig:workflow} shows an (incomplete) set of \LTL formulae
(to be read in conjunction) corresponding to the \bpmn workflow in the
upper half.  The first conjunct on the left means that
$\msf{interview}$ must be the first task to be executed, formula
$\neg \msf{getExms} \Until \msf{optIn}$ in the third conjunct of the
figure means that the academic transcripts cannot be accessed if the
student has opted out. Formula
$\neg \msf{findJobs} \Until (\msf{getExms} \vee \msf{getExp})$ in the
fourth conjunct means that the execution of task $\msf{findJobs}$ must
not happen before the execution of $\msf{getExms}$ or $\msf{getExp}$.
%%%%%%%%%%%%%%%%%%%%%%%%%%%%%%%%%%%%%%%%%%%%%%%%%%%%%%%%%%%%%%%%%%%%%%%%%%
%% (contrast this with the \textsc{bpmn} requires a complex
%% construction with a duplicate task.
%%%%%%%%%%%%%%%%%%%%%%%%%%%%%%%%%%%%%%%%%%%%%%%%%%%%%%%%%%%%%%%%%%%%%%%%%%

The second reason for chosing \LTL to specify workflows is that it
allows us to derive a precise semantics of purpose-aware policies and
to reuse available techniques for the off-line and on-line
verification of formulae to support the analysis of policies at
design-time and their enforcement at run-time.  Such verification
tasks are presented in Section~\ref{sec:verification}.  Here we focus
on the semantics of purpose-aware policies, starting with the
meaning of \LTL formulae, which are expressions of the following
grammar:
%%%%%%%%%%%%%%%%%%%%%%%%%%%%%%%%%%%%%%%%%%%%%%%%%%%%%%%%%%%%%%%%%%%%%%%%%%
%% which are specified using
%% %%%%%%%%%%%%%%%
%% %%%%%%%%%%%%%%%
%% \declare, a declarative workflow modeling language originally
%% introduced in \cite{AalstPS09}.
%%
% While most of the languages for specifying business processes are
% based on an procedural approach \cite{Weske-BPMBook}, meaning that
% they explicitly describe all the allowed sequences of tasks (or
% activities), \declare specifies a set of constraints that must not be
% violated during process executions.  In this way, the language enjoys
% flexibility and is suitable for specifying purpose-based policies.
%%
%% \declare describes workflows as a set of contraints that must not be
%% violated during workflow executions. 
% %%
% In order to make the language
% more accessible to non-technical people, among all possible \LTL
% constraints, some specific ``patterns'' have been singled out as
% particularly meaningful for expressing processes and associated to 
% graphical symbols. Those are basically formulas templates such as
% existence, precedence, succession and choice between activities.
% %%
%%
%%%%%%%%%%%%%%%%%%%%%%%%%%%%%%%%%%%%%%%%%%%%%%%%%%%%%%%%%%%%%%%%%%%%%%%%%%
\[\varphi ::= a \mid \lnot \varphi \mid \varphi_1\land
\varphi_2 \mid \Next\varphi \mid \varphi_1\Until\varphi_2 \mid
%\varphi_1 \Wuntil \varphi_2 \mid
\Box \varphi \mid \Diamond \varphi \text{~~~with $a\in\Prop$}\]
%%%%%%%%%%%%%%%%%%%%%%%%%%%%%%%%%%%%%%%%%%%%%%%%%%%%%%%%%%%%%%%%%%%%%%%%%%
where $\Prop$ is a set of Boolean variables representing tasks.
Intuitively, $\Next\varphi$ means that $\varphi$ holds at the
\emph{next} instant, $\varphi_1\Until\varphi_2$ means that at some
future instant $\varphi_2$ will hold and \emph{until} that point
$\varphi_1$ holds,
%%%%%%%%%%%%%%%%%%%%%%%%%%%%%%%%%%%%%%%%%%%%%%%%%%%%%%%%%%%%%%%%%%%%%%%%%%
%% its weak version $\varphi_1 \Wuntil \varphi_2$ says essentially the
%% same but $\varphi_2$ can also hold forever; 
%%%%%%%%%%%%%%%%%%%%%%%%%%%%%%%%%%%%%%%%%%%%%%%%%%%%%%%%%%%%%%%%%%%%%%%%%%
$\Box \varphi$ means that $\varphi$ always holds, and its dual
$\Diamond \varphi$ that $\varphi$ eventually holds.
%%%%%%%%%%%%%%%%%%%%%%%%%%%%%%%%%%%%%%%%%%%%%%%%%%%%%%%%%%%%%%%%%%%%%%%%%%
Since we assume workflows to \emph{eventually terminate}, we adopt the
finite-trace semantics in~\cite{DeGiacomoMGMM14,DeGiacomoDMM14}.  The
only notable aspect of this semantics (with respect to the standard
semantics as given in, e.g.,~\cite{kroger-merz}) is that
$\Next \varphi$ is true iff a next state actually exists \emph{and} it
satisfies $\varphi$.  The models of \LTL formulae are finite sequences
of Boolean assignments to the variables in $\Prop$ indexed over
natural numbers, which represent instants of a linear and discrete
time.  The idea is that at a certain time instant, the Boolean
variable representing a task is assigned to $\mathit{true}$ iff the
corresponding task has been executed.
As customary in workflow specifications, we assume that one task only
is executed at a time.
By looking at a sequence of
Boolean assignments satisfying a formula, we can thus understand which
tasks have been executed and which are not.  In other words, a
sequence $\Pi$ of Boolean assignments satisfying a formula $\varphi$, in
symbols $\Pi \models \varphi$, correspond to a possible execution of the
workflow described by $\varphi$
(see~\cite{DeGiacomoMGMM14,DeGiacomoDMM14} for a precise definition).
The set of all sequences satisfying a formula $\varphi$, i.e.\ the set
of all successful executions of the workflow described by $\varphi$,
is called its \emph{language} and is denoted by
$\mathcal{L}(\varphi)$.
%%%%%%%%%%%%%%%%%%%%%%%%%%%%%%%%%%%%%%%%%%%%%%%%%%%%%%%%%%%%%%%%%%%%%%%%%%
%% Semantics of \LTL formulae defines which are the executions, or
%% \emph{traces}, $\pi=a_1, a_2, \ldots a_n$ (with $a_i \in \Prop$) which
%% satisfy a formula $\varphi$, written $\pi \models \varphi$. Such a set
%% is also called the \emph{language} of $\varphi$, and denoted by
%% $L(\varphi)$. 
%%%%%%%%%%%%%%%%%%%%%%%%%%%%%%%%%%%%%%%%%%%%%%%%%%%%%%%%%%%%%%%%%%%%%%%%%%
%%
% Technically, a trace $\pi$ is a finite word over the
% alphabet of $2^{\P}$, i.e, over interpretations for symbols in
% $\P$. Given a finite trace $\pi$ and a \LTL formula
% $\varphi$, the truth relation $\pi \models \varphi$ is inductively defined,
% and we refer to \cite{DeGiacomoMGMM14,DeGiacomoDMM14} for the
% technical aspects.
%%
%%%%%%%%%%%%%%%%%
% As customary when modeling process constraints \cite{AalstPS09}, we
% assume that a single task is performed at a time. This is formally
% captured by implicitly assuming the additional constraint in processes specification:
% $
% \Box(\bigvee_{a\in\P} a) \land \Box(\bigwedge_{a,b\in\Prop, a \neq b}
% a \limp \lnot b) %\land \Box(\Wnext \false \leftrightarrow \Ended)
% $.
%%
%%%%%%%%%%%%%%%%%%%%%%%%%%%%%%%%%%%%%%%%%%%%%%%%%%%%%%%%%%%%%%%%%%%%%%%%%%
It is possible to build an \emph{automaton} (i.e., a finite-state
machine) from a formula $\varphi$ accepting exactly all the traces
belonging to $\mathcal{L}(\varphi)$; see
again~\cite{DeGiacomoMGMM14,DeGiacomoDMM14} for the description of the
procedure for doing this.

We can now define the notion of purpose-aware policy as a tuple
\[\begin{array}{lcl}
\P & = & (\mathit{DataOwners}, \mathit{Subjects}, \mathit{DataObjects},
                             \mathit{Actions}, \mathit{Tasks}, \\
    ~& ~& \mathit{Workflows}, \mathit{Purposes}, \mathit{dcp}, \mathit{rcp}, 
                             \mathit{sod}, \mathit{bod},
          \mathit{cont}, \mathit{def})
  \end{array} \]
whose components are as explained above.
%%%%%%%%%%%%%%%%%%%%%%%%%%%%%%%%%%%%%%%%%%%%%%%%%%%%%%%%%%%%%%%%%%%%%%%%%%%%%%
%%% SR: TO BE REUSED SOMEWHERE ELSE???? %%%%%%%%%%%%%%%%%%%%%%%%%%%%%%%%%%%%%%
%%%%%%%%%%%%%%%%%%%%%%%%%%%%%%%%%%%%%%%%%%%%%%%%%%%%%%%%%%%%%%%%%%%%%%%%%%%%%%
%% Notice that, from a practical viewpoint, $\P$ is simply a set of
%% tuples that can be stored in a repository, say, a relational database.
%% %%
%% Also, observe that policies share specific sets: purpose-centric
%% policies shares $Purposes$ with $dcp$ and $Actions$ with $rcp$ and
%% they all share $DataObjects$. This allows for joining them and
%% verifying the complex properties required for our verification tasks.
%%%%%%%%%%%%%%%%%%%%%%%%%%%%%%%%%%%%%%%%%%%%%%%%%%%%%%%%%%%%%%%%%%%%%%%%%%%%%%

As it is standard in access control models, we introduce the notion of
a request
%%%%%%%%%%%%%%%%%%%%%%%%%%%%%%%%%%%%%%%%%%%%%%%%%%%%%%%%%%%%%%%%%%%%%%%%%%%%%%
%% Before giving the formal semantics to policies, we describe the last
%% conceptual component of our framework: requests. Requests are the
%% \emph{events} happening in our system and represent the will of a
%% subject to perform a task over the data of a specific data owner for a
%% purpose. 
%%%%%%%%%%%%%%%%%%%%%%%%%%%%%%%%%%%%%%%%%%%%%%%%%%%%%%%%%%%%%%%%%%%%%%%%%%%%%%
as a tuple %% Formally, requests are tuples of relation of the form
$(\mathit{wid}, \mathit{sub}, \mathit{tsk}, \mathit{do}, \mathit{p})$
where $\mathit{sub}\in \mathit{Subjects}$, $\mathit{tsk}\in
\mathit{Tasks}$, $\mathit{do}\in \mathit{DataOwners}$, $\mathit{p}\in
\mathit{Purposes}$, and $\mathit{wid}$
%%%%%%%%%%%%%%%%%%%%%%%%%%%%%%%%%%%%%%%%%%%%%%%%%%%%%%%%%%%%%%%%%%%%%%%%%%%%%%
%% $req
%% \subseteq Wid \times Subjects \times Tasks \times DataOwner \times
%% Purposes$, where 
%%%%%%%%%%%%%%%%%%%%%%%%%%%%%%%%%%%%%%%%%%%%%%%%%%%%%%%%%%%%%%%%%%%%%%%%%%%%%%
belongs to the set $\mathit{Wid}$ of workflow identifiers (allowing us
to distinguish among different executions of possibly the same
workflow).  Intuitively, $(\mathit{wid}, \mathit{sub}, \mathit{tsk},
\mathit{do}, p)$ means that subject $\mathit{sub}$ asks the permission
to execute task $\mathit{tsk}$ on the data objects owned by the data
owner $\mathit{do}$ in the workflow instance $\mathit{wid}$ for the
purpose $p$.  The relation $\mathit{req} \subseteq \mathit{Wid} \times
\mathit{Subjects} \times \mathit{Tasks} \times \mathit{DataOwner}
\times \mathit{Purposes}$ contains all possible requests.

\subsection{Semantics of purpose-aware policies}
\label{subsec:semantics}
We explain how a request $(\mathit{wid}, \mathit{sub}, \mathit{tsk},
\mathit{do}, p)$ is granted or denied according to a purpose-aware
policy $\P$.  The idea is to derive a first-order \LTL formula from
the \LTL formula $\mathit{def}(p)$ constraining the ordering of
requests such that the rule-, data-centric
policies and \sod and \bod constraints are satisfied.  Thus,
instead of sequences of Boolean assignments, we consider first-order
models which differ for the interpretation of requests only.  For
the sake of brevity, we do not give a formal semantics of first-order
\LTL on finite-traces but only some intuitions and refer the
interested reader to~\cite{DeMasellisS13} for the details.
%%%%%%%%%%%%%%%%%%%%%%%%%%%%%%%%%%%%%%%%%%%%%%%%%%%%%%%%%%%%%%%%%%%%%%%%
%% Observe that requests play the same role of (propositional) tasks
%% in \declare, but they carry a payload with information needed to
%% check policies.
%%%%%%%%%%%%%%%%%%%%%%%%%%%%%%%%%%%%%%%%%%%%%%%%%%%%%%%%%%%%%%%%%%%%%%%%
%% Indeed, as traditional workflow traces are a sequences of tasks
%% $\pi=a_1,a_2, \ldots, a_n$, in our system traces are sequences of
%% requests $\Pi=req_1(\ldots), req_2(\ldots) \ldots req_n(\ldots)$,
%% possibly accounting interleaved requests of different workflows being
%% executed in parallel. 
%%%%%%%%%%%%%%%%%%%%%%%%%%%%%%%%%%%%%%%%%%%%%%%%%%%%%%%%%%%%%%%%%%%%%%%%

First of all, we observe that every workflow instance can be
considered in isolation since the framework presented above allows one
to specify only constraints within a workflow instance and not accross
instances.  For this reason, we introduce an operator to identify
requests referring to the same workflow instance $\msf{wid}$ out of a
trace $\Pi$ containing requests referring to arbitrary workflow
instances, i.e.\ $\Pi|_{\msf{wid}} = req_1(\msf{wid},\mathit{sub}_1,
\mathit{tsk}_1, \mathit{do}_1, \msf{p}),
%%%%%%%%%%%%%%%%%%%%%%%%%%%%%%%%%%%%%%%%%%%%%%%%%%%%%%%%%%%%%%%%%%%%%%%%
% req_2(\msf{wid},\mathit{sub}_2, \mathit{tsk}_2, \mathit{do}_2, \msf{p})
%%%%%%%%%%%%%%%%%%%%%%%%%%%%%%%%%%%%%%%%%%%%%%%%%%%%%%%%%%%%%%%%%%%%%%%%
\ldots, req_n(\msf{wid},\mathit{sub}_n, \mathit{tsk}_n, \mathit{do}_n,
\msf{p})$ is the trace representing the evolution of the specific
workflow instance $\msf{wid}$.  Notice that requests in
$\Pi|_{\msf{wid}}$ share the same workflow identifier $\mathsf{wid}$
and purpose $\mathsf{p}$ whereas subjects, tasks, and data owners may
be different.
%%%%%%%%%%%%%%%%%%%%%%%%%%%%%%%%%%%%%%%%%%%%%%%%%%%%%%%%%%%%%%%%%%%%%%%%
%%
%%
%% We formally define the semantics of our policies in terms of sequences
%% of requests that satisfy it. 
%%%%%%%%%%%%%%%%%%%%%%%%%%%%%%%%%%%%%%%%%%%%%%%%%%%%%%%%%%%%%%%%%%%%%%%%
Given a purpose-aware policy $\P$, for each purpose $\msf{p} \in
Purposes$ such that $def(\msf{p}) = \varphi$, we build a (first-order)
\LTL formula $\Phi_{\msf{p}}:=\varphi\wedge \Lambda\wedge \Sigma\wedge
B$ where
\begin{align*}
  \Lambda := &
   \bigwedge_{\mathit{cont}(\varphi,t)} 
    t \leftrightarrow 
    \exists sub,do.
    \left(
    \begin{array}{l}
      req(\msf{wid}, sub, t, do, \msf{p}) \land \\
     \bigwedge_{uses(act, t, obj)} 
      dcp(do, obj,\msf{p}) \land rcp(sub, act, obj)
    \end{array}
    \right)  \\
   \Sigma :=& \bigwedge_{sod(\varphi, t_1, t_2)} \xi(t_1,t_2,=)\wedge \xi(t_2,t_1,=) 
    \\
   B :=& \bigwedge_{bod(\varphi, t_1, t_2)} \xi(t_1,t_2,\neq)\wedge \xi(t_2,t_1,\neq)  \\
   \xi(t,t',\bowtie)  :=   &
    \Box \forall sub,sub'.\left(
    \begin{array}{l}
       \forall do.req(\msf{wid}, sub, t, do, \msf{p}) \land \\
         \Diamond \forall do.req(\msf{wid}, sub', t', do, \msf{p}) 
    \end{array} 
    \right) \ra sub \bowtie sub' \, .
%%%%%%%%%%%%%%%%%%%%%%%%%%%%%%%%%%%%%%%%%%%%%%%%%%%%%%%%%%%%%%%%%%%%%%%%
%% && \land \\
%% && \Box \forall
%%   sub_1,sub_2.(\forall do.req(wid, sub_1, t_2, do, \msf{p}) \land \\
%%   && \Diamond \forall do.req(wid, sub_2, t_1, do, \msf{p}) \ra sub_1 \not = sub_2
%% )) \land \\ \\
%%
%% (4)& \bigwedge_{t_1, t_2 \in bod(\varphi, t_1, t_2)} & (\Box \forall
%%   sub_1,sub_2.( \forall do.req(wid, sub_1, t_1, do, \msf{p}) \land \\
%%   && \Diamond \forall do. req(wid, sub_1, t_2, do, \msf{p}) \ra sub_1 = sub_2)
%% ) \\
%% && \land \\
%% && (\Box \forall
%%   sub_1,sub_2.(\forall do.req(wid, sub_1, t_2, do, \msf{p}) \land \\
%%   && \Diamond \forall do.req(wid, sub_2, t_1, do, \msf{p}) \ra sub_1 = sub_2
%% )))
%%%%%%%%%%%%%%%%%%%%%%%%%%%%%%%%%%%%%%%%%%%%%%%%%%%%%%%%%%%%%%%%%%%%%%%%
\end{align*}
Formula $\Lambda$ says that in order to execute task $\msf{t}$ for
purpose $\msf{p}$, we need to check that subject $\mathit{sub}$ who
has requested to execute it is entitled to do so according to both the
rule- and data-centric policies in $\P$\boldred{, thus formalizing the
  check (\textbf{C1}) in the introduction}.  Formulae $\Sigma$ and $B$
encode the \sod and \bod constraints in $\P$, respectively, which are
both derived from the same template formula $\xi(t,t',\bowtie)$,
saying that if a request for executing $t'$ is seen after that for
executing $t$, then the two subjects performing such tasks must be
either different (when $\bowtie$ is $\neq$, i.e., in case of a \sod)
or equal (when $\bowtie$ is $=$, i.e., in case of a \bod). \boldred{Formulae
$\varphi$, $\Sigma$ and $B$, thanks to their temporal
characterization, formalize the check (\textbf{C2}) in the
introduction.}
%%%%%%%%%%%%%%%%%%%%%%%%%%%%%%%%%%%%%%%%%%%%%%%%%%%%%%%%%%%%%%%%%%%%%%%%%%%%
%% \begin{compactitem}
%% \item[$(1)$] constraints the execution of tasks being the workflow
%%   $\varphi$ for $\msf{p}$;
%% \item[$(2)$] defines tasks in $\varphi$ as requests and verifies that
%%   data- and rule-centric policies are satisfied. Precisely, each
%%   propositional event (task) $\msf{t}$ in $\varphi$ corresponds to a
%%   request $req(wid, sub, \msf{t}, do, \msf{p})$ in our system, where
%%   the subject $sub$ performing the request and the data owner $do$ are
%%   existentially quantified, and are such that for all action $act$ task
%%   $\msf{t}$ uses over some data object $obj$: \myi $do$ has allowed the
%%   access of $obj$ for purpose $\msf{p}$ ($dcp$) and \myii $sub$ is authorized
%%   to perform $act$ over $obj$ ($rcp$).
%% \item[$(3)$] takes care of \sod constraints. For each $sod(\varphi,
%%   t_1, t_2) \in \P$, the formula says that if a request for executing $t_2$
%%   is seen after that for executing $t_1$ (or vice-versa), then the two
%%   subjects performing them must be different.
%% \item[$(4)$] takes care of \bod constraints analogously.
%% \end{compactitem}
%%%%%%%%%%%%%%%%%%%%%%%%%%%%%%%%%%%%%%%%%%%%%%%%%%%%%%%%%%%%%%%%%%%%%%%%%%%%
%% Formally, the variable $\msf{wid}$ in $\Phi_{\msf{p}}$ is \emph{free}
%% (i.e.\ not bound by any quantifier) and this means that it must be
%% substituted with an actual value in order for the formula to be
%% evaluated.  
%%%%%%%%%%%%%%%%%%%%%%%%%%%%%%%%%%%%%%%%%%%%%%%%%%%%%%%%%%%%%%%%%%%%%%%%%%%%
A sequence of requests $\Pi|_{\msf{wid}}$ for a purpose $\msf{p}$ and
an instance $\msf{wid}$ of the workflow $\mathit{def}(\msf{p})$
\emph{satisfies} the purpose-aware policy $\P$ iff $\Pi|_{\msf{id}}
\models \Phi_{\msf{p}}$. 
%%%%%%%%%%%%%%%%%%%%%%%%%%%%%%%%%%%%%%%%%%%%%%%%%%%%%%%%%%%%%%%%%%%%%%%%%%%%
%% , where $\Phi_{\msf{p}}[\msf{id}/wid]$ is the formula obtained from
%% $\Phi_{\msf{p}}$ by substituting the workflow identifier $\msf{id}$
%% for $wid$.
%%%%%%%%%%%%%%%%%%%%%%%%%%%%%%%%%%%%%%%%%%%%%%%%%%%%%%%%%%%%%%%%%%%%%%%%%%%%
By abusing notation, we write $\mathcal{L}(\varphi)$ for all such
sequences.  Given a sequence $\sigma$ of (previous) requests, a (new)
request
$r=(\msf{wid}, \mathit{sub}, \mathit{tsk}, \mathit{do}, \msf{p})$ is
\emph{granted} by the purpose-aware policy $\P$ iff $\msf{wid}$ is an
instance of the workflow $\mathit{def}(\msf{p})$ and there exists a
sequence $\sigma'$ of requests such that $\sigma, r,\sigma'$ is in
$\mathcal{L}(\varphi)$; otherwise, it is \emph{denied}.

To illustrate some of the notions introduced above, let us consider
the first-order \LTL formula that can be derived from the example in
Section~\ref{sec:motivations}.  As already observed, the formula
$\varphi$ associated to the purpose $\msf{JobHunting}$ is the
conjunction of the formulae in the lower half of
Figure~\ref{fig:workflow}.  The conjunct in $\Lambda$ for
$\msf{interview}$ is
\begin{eqnarray*}
\msf{interview} &\leftrightarrow & 
 \exists sub,do.\left(
 \begin{array}{l}
   req(\msf{wid},sub, \msf{interview}, do, \msf{jobHunting}) \land\\
   dcp(do, \msf{userProfile},\msf{jobHunting}) \land \\
   rcp(sub, \msf{read},\msf{userProfile})
 \end{array}
 \right)\, .
\end{eqnarray*}
The formula representing the \sod constraint between $\msf{interview}$
and $\msf{findJobs}$ is
\begin{eqnarray*}
  \Box \forall sub,sub'.\left(
    \begin{array}{l}
       \forall do.req(wid, sub, \msf{interview}, do, \msf{jobHunting}) \land \\
         \Diamond \forall do.req(wid, sub', \msf{findJobs}, do, \msf{jobHunting}) 
    \end{array} 
    \right) \ra sub \neq sub' & \wedge \\
  \Box \forall sub,sub'.\left(
    \begin{array}{l}
       \forall do.req(wid, sub, \msf{findJobs}, do, \msf{jobHunting}) \land \\
         \Diamond \forall do.req(wid, sub', \msf{interview}, do, \msf{jobHunting}) 
    \end{array} 
    \right) \ra sub \neq sub' & .
\end{eqnarray*}
Notice that the second conjunct above can be dropped without loss of
generality since, from $\mathit{def}(\msf{jobHunting})$, it is
possible to derive that it is never the case that task
$\msf{findJobs}$ is executed before task $\msf{interview}$.  From a
complete specification of the purpose-aware policy $\P$ for the
running example, it is not difficult to see that the request
$r_0=(\msf{wid},\msf{bob}, \msf{interview}, \msf{sam},
\msf{jobHunting})$ is granted by applying the definition given above
as follows.  First, take $\sigma$ to be the empty sequence as $r_0$ is
the first request.  Second, we can derive that task $\msf{interview}$
can be executed from the formula for $\msf{interview}$ above and the
fact that $\P$ is such that $\mathit{dcp}(\msf{sam},
\msf{userProfile},\msf{jobHunting})$, i.e.\ $\msf{sam}$ decided to
release his user profile for the purpose of job hunting, and
$\mathit{rcp}(\msf{bob}, \msf{read},\msf{userProfile})$,
i.e.\ $\msf{bob}$ has the right to read user profiles.  Third, take
$\sigma'=r_1, r_2, r_3, r_4, r_5$ for
%%%%%%%%%%%%%%%%%%%%%%%%%%%%%%%%%%%%%%%%%%%%%%%%%%%%%%%%%%%%%%%%%%%%%%%%
{\small
\begin{align*}
 r_1 := (\msf{wid},\msf{sam}, \msf{optOut}, \msf{sam}, \msf{jobHunting}) & &
 r_2 := (\msf{wid},\msf{bob}, \msf{getExp}, \msf{sam}, \msf{jobHunting})\\
 r_3 := (\msf{wid},\msf{adam}, \msf{findJobs}, \msf{sam}, \msf{jobHunting}) & &
 r_4 := (\msf{wid},\msf{bob}, \msf{propJobs}, \msf{sam}, \msf{jobHunting}) \\
 r_5 := (\msf{wid},\msf{sam}, \msf{choosJob}, \msf{sam}, \msf{jobHunting}) & &
\end{align*}}
%%%%%%%%%%%%%%%%%%%%%%%%%%%%%%%%%%%%%%%%%%%%%%%%%%%%%%%%%%%%%%%%%%%%%%%%
where $r_1$ corresponds to the fact that $\msf{sam}$ has opted out and
only his past experiences can be released, $r_2$ to the fact that
$\msf{bob}$ can retrive $\msf{sam}$'s past experiences (as said in
Section~\ref{sec:motivations}), $r_3$ to the fact that the jobs for
$\msf{sam}$ are found by $\msf{adam}$, who is distinct from
$\msf{bob}$ in order to satisfy the \sod constraint between
$\msf{interview}$ and $\msf{findJobs}$ in Figure~\ref{fig:workflow}
(indeed, we assume that $\msf{bob}$ can execute $\msf{getExp}$
according to $\P$), $r_4$ to the fact that the list of found jobs is
proposed to $\msf{sam}$ by $\msf{bob}$ in order to satisfy the \bod
constraint between $\msf{interview}$ and $\msf{propJobs}$ in
Figure~\ref{fig:workflow}, and $r_5$ to the fact that $\msf{sam}$
decides to pick a job from the proposed list.
% Indeed, the sequence
% $\sigma'$ just described corresponds to one of the possible future
% evaluations of the workflow instance $\msf{wid}$.  Its existence
% implies that it is possible to grant the request $r_0$ since the
% workflow instance $\msf{wid}$ can be succesfully terminated.  When
% the actual next request $\widehat{r_1}$ happens, $\sigma'$ may be
% useful or not depending on the fact that $\widehat{r_1}=r_1$ or not.

%%%%%%%%%
We close the Section by remarking that this technique allows to
discover inconsistencies as soon as they occur, i.e., at the earliest
possible time. This is sometimes called \emph{early detection} in the
\bpm literature, and it is a notable feature of temporal logics. To
explain the concept, assume that, according to $\P$, \emph{only}
$\msf{bob}$ can perform the activities of $\msf{jobHunting}$: when
$r_0=(\msf{wid},\msf{bob}, \msf{interview}, \msf{sam},
\msf{jobHunting})$
(or actually any other request for purpose $\msf{jobHunting}$) is
presented to the system, we are able to understand that no execution
can ever successfully complete the workflow, as sooner or later task
$\msf{findJobs}$ must be executed by someone different from
$\msf{bob}$ which however does not have the rights to do it. As a
result, $r_0$ is denied and hence $\msf{bob}$ is not granted to access
the data even if, by observing the current state, there is no evidence
yet of any violation.
%%%%%%%%%

The
decidability and complexity of answering requests is studied in the
following section.

\section{Policies verification}
\label{sec:verification}

We now formalize, provide a solution and give complexity results to
the following verification tasks:
\begin{compactitem}
\item[\emph{Purpose achievement problem}]: given a purpose-aware
  policy $\P$ and a purpose $\msf{p}$, is it possible to successfully
  execute workflow $\mathit{def}(\msf{p})$? That is to say: is it
  possible to assign tasks to subjects such that policy $\P$ is
  satisfied and the workflow successfully terminates?

\item[\emph{Runtime policy enforcement}]: given a purpose-aware policy
  $\P$, the current workflow execution trace $\pi$, and a new request
  $r_i$, can $r_i$ be granted or must be denied in order for the
  workflow to eventually terminate \boldred{(check (\textbf{C2}) in the
  introduction)} and such that the sequence of
  granted request always satisfies $\P$ \boldred{(check (\textbf{C1}))}?
%%
% ? In making this choice
%   we consider that: \myi granting $r_i$ should not jeopardize the successful
%   termination of the workflow and \myii the sequence of allowed
%   request must always satisfy $\P$.
\end{compactitem}

% We propose two basic kind of policies verification, which address
% different offline and runtime needs: the \emph{offline workflow
%   satisfiability verification} and the \emph{runtime policies
%   enforcement}.

% As customary in when dealing with temporal logics, given a formula
% $\varphi$ representing the constraints to be checked/verified, the
% construction of an automaton is needed. Depending on the kind of
% verification (offline or runtime) to be performed, a specific
% automaton is built, and used in a different way.
% %%
% Notably, our constraint formulas are first-order, and
% hence the traditional techniques to build the automaton cannot be used
% as they are.
% %%
% However, the modularity of our framework and the use of 
% symbolic techniques allow to build the automaton with a reasonable
% complexity. We show how to build such automata and then how to use
% them for our offline and runtime verification tasks.

%%%%%%%%%%%%%%%%%%%%%%%%%%%%%%%%%%%%%%%%%%%%%%%%%%%%%

\subsection{Purpose achievement problem}

Technically speaking, this problem amounts to check, given a purpose
$\msf{p}$ in a purpose-aware policy $\P$, if $\Phi_{\msf{p}}$ is
satisfiable, i.e., if there exists a trace $\Pi|_{\msf{wid}}$ (where
$\msf{wid}$ is a generic identifier for workflow
$\mathit{def}(\msf{p})$) such that
$\Pi|_{\msf{wid}} \models \Phi_{\msf{p}}$.

We adopt an automata-based approach to solve the problem, which
consists in building the automaton for $\Phi_{\msf{p}}$ and check if
there exists a path to a final state.
Since $\Phi_{\msf{p}}$ is first-order, we exploit the modularity of
our framework and the of symbolic techniques in \cite{DeMasellisS13}
to build the automaton with a reasonable complexity.
Indeed, we notably separate sub-formulas in $\Phi_{\msf{p}}$ checking
the workflow ($\varphi$), checking $rcp$/$dcp$ ($\Lambda$) and
checking \sod/\bod constraints ($\Sigma$,$B$). Moreover, we observe
that while $\varphi$ and $\Sigma$,$B$ typically remain fixed as they
describe the workflow, $rcp$ and $dcp$ policies may occasionally
change.
\begin{figure}[t]
  \centering
  \includegraphics[scale=0.8]{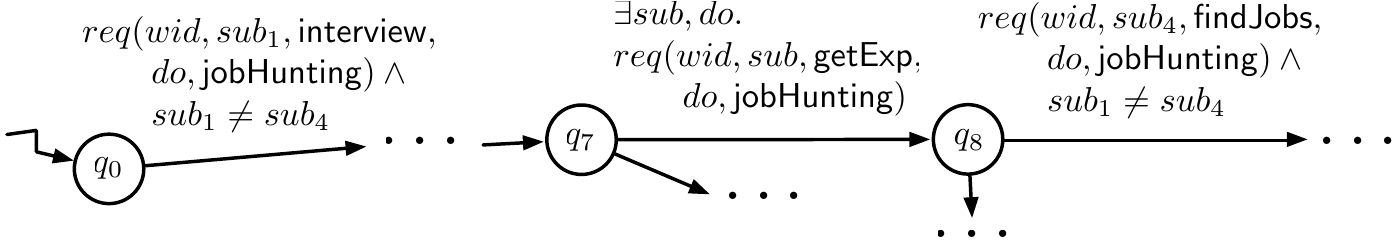}
  \caption{An excerpt of the pre-automaton for formula
    $\Phi_{\msf{jobHunting}}$.}
  \label{fig:exampleAutomaton}
\end{figure}
From these considerations, we decouple the construction of the
automaton checking the workflow constraints only, which we call
\emph{pre}-automaton, from that taking also care of $rcp$ and $dcp$,
called \emph{specialized}-automaton, so as when $rcp$ and $dcp$
change, we do not need to re-compute the automaton from scratch, but
we just ``adapt'' the already computed pre-automaton to the new
policies.

The pre-automaton is a finite-state machine with edges labeled with
first-order formulas. Figure~\ref{fig:exampleAutomaton} shows an
excerpt of the pre-automaton for $\Phi_{\msf{jobHunting}}$, where we
notice that variables for subjects involved in \sod or \bod
constraints are \emph{free}, i.e., not bounded by any
quantifier. Indeed, as tasks $\msf{interview}$ and $\msf{jobHunting}$
must be executed by different subjects (\sod), variables for such
subjects, namely $sub_1$ and $sub_4$, are free and must be different.

% Given a purpose-aware policy $\P$ and a purpose $\msf{p}$, the
% \emph{pre-automata} $\ol{A}_{\msf{p}}$ is build from formula
% $\Phi_{\msf{p}}(wid)$ by ignoring constraints
% $\bigwedge_{uses(act, \msf{t}, obj)} (dcp(do, obj, \msf{p}) \land
% rcp(sub, act, obj))$.

\begin{theorem}
  Given a purpose-aware policy $\P$, and a purpose $\msf{p}$ with
  $\mathit{def}(\msf{p})=\varphi$, the construction of the
  \emph{pre-automata} $\ol{A}_{\msf{p}}$ requires exponential space in
  the number of temporal operators of $\varphi$, $sod$ and $bod$.
  \label{th:automaton}
\end{theorem}
\begin{proofsk}
  Automaton $\ol{A}_{\msf{p}}$ is built from $\Phi_{\msf{p}}$ using
  the algorithm for the propositional case over finite traces
  \cite{DeGiacomoMGMM14}, by simply ignoring constraint
  $\bigwedge_{uses(act, t, obj)} dcp(do, obj,\msf{p}) \land rcp(sub,
  act, obj)$
  and quantifiers $\forall sub_1, sub_2$ of sub-formulas $\Sigma$ and
  $B$ and considering first-order atoms as propositional symbols (see
  \cite{DeMasellisS13} for details). Observing that $\varphi$,
  $\Sigma$ and $B$ are the only sub-formulas which contains temporal
  operators, and hence, the only source of complexity, and recalling
  that building the automaton for a propositional formula requires
  exponential space in the number of temporal operators
  \cite{DeGiacomoMGMM14}, i.e., $2^{|\varphi, bod, sod|}$ in the worst
  case, the claim follows.
\end{proofsk}

A pre-automaton can be specialized to take into account a specific set
of $rcp$ and $dcp$ by simply adding, to each edge labeled with
$req(wid, sub, \msf{t}, do, \msf{p})$, the constraint
$\bigwedge_{act,obj \in uses(act, \msf{t}, obj)} (dcp(do, obj,
\msf{p}) \land rcp(sub, act, obj))$.
Notice that such an operation is legitimate as sub-formulas
($\Lambda$) checking $rcp$ and $dcp$ do not contain any temporal
operators. The resulting specialized-automaton $A_{\msf{p}}$ is a
symbolic structure where we check whether there is a way to reach a
final state---thus solving a workflow satisfiability problem---by
trying to satisfy formulas on edges. Indeed, satisfy a formula
(w.r.t. $\P$) precisely means assigning a task to a subject which is
authorized by $\P$ to perform it. Consider, e.g., formula
$\exists sub, do.req(wid, sub, \msf{findJobs}, do, \msf{jobHunting})$
on edge from $q_7$ to $q_8$ in Figure~\ref{fig:exampleAutomaton}: its
corresponding formula in the specialized automaton is
$\exists sub, do.req(wid, sub, \msf{findJobs}, do, \msf{jobHunting})
\bigwedge_{uses(act, \msf{findJobs}, obj)} dcp(do, obj,\msf{p}) \land
rcp(sub, act, obj)$.
Assuming that $\msf{bob}$ has the rights to perform $\msf{getExp}$,
the substitution $\msf{bob}/sub$ satisfy the formula according to
$\P$, and so that edge can be use to build a path to a final
state. Analogously, the assignment $\msf{bob}/sub_1$ and
$\msf{adam}/sub_4$ satisfies formulas on edges from $q_0$ and $q_8$
respectively. When no such an assignment can be found, the workflow
cannot be successfully completed given policy $\P$.

\begin{theorem}
  Given a finite purpose-aware policy $\P$ and a purpose $\msf{p}$
  with $\mathit{def}(\msf{p})=\varphi$ the \emph{purpose achievement
    problem} can be solved in exponential time in the size of
  $\varphi$, $sod$ and $bod$.
\label{th:WSP-complexity}
\end{theorem}
\begin{proof}
  Given a specialized automaton $A_{\msf{p}}$, in the worst case all
  possible assignments to variables in edges must be checked. Such
  assignments are finite as the purpose $\P$ is finite and exponential
  in the number of variables. Actually, it can be proven (see
  \cite{DeMasellisS13}) that it is enough to try substitutions for
  subjects involved in \sod and \bod constraints only, which usually
  are of a much smaller number. This is because \sod and \bod
  constraints generate variables which occur throughout the whole
  automaton (across-state variables) and are the only source of
  complexity (such as $sub_1$ and $sub_4$ in
  Figure~\ref{fig:exampleAutomaton}). Any other variable, instead, is
  local, i.e., it occurs in a specific formula only, whose formula can
  be checked for satisfiability by simply querying $\P$.
  For each substitution to across-state variables a reachability test
  to a final state (linear in the size of the automaton) must be
  performed. From Theorem~\ref{th:automaton}, we get time complexity
  of $|Subjects|^{|sod, bod|} \cdot 2^{|\varphi, bod, sod|}$.
\end{proof}

%%% Local Variables:
%%% mode: latex
%%% TeX-master: "main"
%%% save-place: t
%%% End:

%!TEX root = ./main.tex

\subsection{Runtime policies verification}

Given a sequence $\pi$ of (previous) requests and a new request $r$,
should we allow $r$ or not?

Traditional \LTL semantics presented in the previous Section is not
adequate for evaluating requests at runtime, as it considers the trace
$\Pi$ seen so far to be \emph{complete}. Instead, we want to evaluate
the current request by considering that the execution could still
continue and this evolving aspect has a significant impact on the
evaluation: at each step, indeed, the outcome may have a degree of
uncertainty due to the fact that future executions are yet unknown.

Consider, e.g., that so far request $r_0$ has been granted, where
$r_0:=(\msf{wid},\msf{bob}, \msf{interview}, \msf{sam},
\msf{jobHunting})$
is as in the previous Section. Assume that now request
$r_1 := (\msf{wid},\msf{sam}, \msf{optOut}, \msf{sam},
\msf{jobHunting})$
is presented and must be evaluated: we may be tempted to use the
traditional \LTL semantics, which returns
$\pi: r_0, r_1 \not \models \Phi_{\msf{p}}$ because some constraints
have not been satisfied (such as: ``eventually $\msf{findJobs}$ must
be executed'') but this is not a good reason to deny request $r_1$, as
the workflow execution will not stop after $r_1$ (and actually any
other single request different from $r_1$ would not have satisfied
$\Phi_{\msf{p}}$ anyway).

A more complex analysis is hence required, which assesses the
capability of a partial trace to satisfy or violate a formula
$\varphi$ \emph{in the future} by analyzing whether it belongs to the
set of \emph{prefixes} of $\L(\varphi)$ and/or the set of prefixes of
$\L(\neg \varphi)$. Roughly speaking, let $\pi: r_0, r_1$ be as before:
we want to check if there exists a possible sequence of future
requests $\pi': r_2, r_3, \ldots r_n$ such that
$\pi,\pi' \models \Phi_{\msf{p}}$ and, if this is the case, we grant
request $r_1$. We can actually be more precise, and evaluate the
current request in four different ways.

Given a (partial) trace $\pi$,
a formula $\Phi_{\msf{p}}$ and a new request $r$, we adopt the runtime
semantics in \cite{DeGiacomoMGMM14} which is such that:
\begin{compactitem}
\item $\pi, r \models [\Phi_{p}]_{\rv}=\temptrue$, when $\pi, r$
  \emph{temporarily satisfies} $\Phi_{p}$, i.e., $\pi$ is currently
  compliant with $\Phi_{p}$, but a possible system future prosecution
  may lead to falsify $\Phi_{p}$;
\item $\pi,r \models [\Phi_{p}]_{\rv}=\tempfalse$, when that the current trace \emph{temporarily
    falsify} $\Phi_{p}$, i.e., $\pi,r$ is not current compliant with
  $\Phi_{p}$, but a possible system future prosecution may lead to
  satisfy $\Phi_{p}$;
\item $\pi,r \models [\Phi_{p}]_{\rv}=\true$, when $\pi,r$ \emph{satisfies} $\Phi_{p}$ and it will
  always do, no matter how it proceeds;
  %% c
\item $\pi,r \models [\Phi_{p}]_{\rv}=\false$, when $\pi,r$
  \emph{falsifies} $\Phi_{p}$ and it will always do, no matter how it
  proceeds.
\end{compactitem}
A new request $r$ is \emph{denied} if
$\pi,r \models [\Phi_{p}]_{\rv}=\false$, and \emph{granted} otherwise.

\boldred{Intuitively, every time a new request is presented, we check
  that: \myi from the current automaton state there exists an outgoing
  edge whose formula is satisfied by the current request (which
  corresponds to check (\textbf{C1}) in the introduction) and \myii
  from the arrival state there exists a path to a final state, which
  is a \wsp that exactly corresponds to check (\textbf{C2}).}
Such analyses are performed on an automaton which is different
from the one presented in the previous Section, as it not only has to
check prefixes of $\Phi_{\msf{p}}$, but also that of
$\neg \Phi_{\msf{p}}$, in order to distinguish among the four cases above
(see \cite{DeGiacomoMGMM14} for details). However, the same idea of
the pre- and specialized automata also apply to this case.

Once the specialized-automaton has been computed, the current sequence
of requests is analyzed. Notice that, differently from the offline
verification, assignment of users to tasks are partially given by the
current and previous requests, and hence we have to check if such
partial assignments can be extended, according to policy $\P$, in order
to reach a final state.
Consider again
$r_0:=(\msf{wid},\msf{bob}, \msf{interview}, \msf{sam},
\msf{jobHunting})$
to be the already granted request and
$r_1 := (\msf{wid},\msf{sam}, \msf{optOut}, \msf{sam},
\msf{jobHunting})$
to be the current one. Request $r_0$ provides the assignment
$\msf{bob}/sub_1$ which forces us to solve the \wsp with the additional
constraint of $sub_1=\msf{bob}$. Actually,
$r_0,r_1 \models [\Phi_{p}]_{\rv}=\tempfalse$, as
$r_0,r_1 \not \models \Phi_{\msf{p}}$ but there exists a sequence of
assignments of users to tasks that eventually satisfies it, which is
sequence $r_2, r_3, r_4, r_5$ shown in the previous Section.  We
remark that a \wsp instance must be performed each time a new request
is presented. Indeed, the actual next request $\hat{r}_2$ (still
unknown at the current time) is in general different from $r_2$, and
hence sequence $r_2, r_3, r_4, r_5$ found at this step as a witness of
a possible future execution is of no use as the system progresses.
Notably, the fact that we discover inconsistencies at the earliest
possible time (early detection) allows us to block workflow executions
exactly when a possible successful path can still be followed. When
this is not guaranteed, the online enforcement is inefficient as when
the precise point of deviation from the right path is unknown: \myi
possible several tasks are executed before realizing the inconsistency
(hence several data are accessed thus breaking the security) and \myii
possibly it is too late to recover the execution.

\begin{theorem}
  \label{thm:main}
  Given a purpose-aware policy $\P$ and a purpose $\msf{p}$ with
  $\mathit{def}(\msf{p})=\varphi$, the \emph{runtime policy
    verification} requires, at each step, exponential time in the size
  of $\varphi$, $sod$ and $bod$.
\end{theorem}

\begin{proofsk}
  The construction of the specialized automaton for monitoring
  incoming requests for workflow $\mathit{def}(\varphi)$ require the
  same complexity of the automaton in Theorem~\ref{th:automaton} (see
  \cite{DeGiacomoMGMM14} and \cite{DeMasellisS13}). When a new request
  $r$ is presented, a \wsp instance must be solved, as explained in
  \cite{DeMasellisS13}. Being the complexity of \wsp as in
  Theorem~\ref{th:WSP-complexity}, the claim follows.
\end{proofsk}

%%% Local Variables:
%%% mode: latex
%%% TeX-master: "main"
%%% save-place: t
%%% End:

\section{Discussion and Related Work}
\label{sec:related}

%%%%%%%%%%%%%%%%%%%%%%%%%%%%%%%%%%%%%%%%%%%%%%%%%%%%%%%%%%%%%%%%%%%%%%%%%%
%% \begin{compactitem}
%% \item \cite{PetkovicPZ11} very similar to ours, same concept of
%%   purpose as process, but no runtime monitor.
%% \item \cite{ByunL08} purpose is just a tree-like hierarchy and is
%%   granted only if the current purpose or one of its ancestors is
%%   declared by the data owner.
%% \item \cite{JafariEtal:ACM2014} cumbersome logic with modalities. No
%%   link to data fields. Purpose = actions (or activities). There are
%%   also workflows.
%% \end{compactitem}
%%%%%%%%%%%%%%%%%%%%%%%%%%%%%%%%%%%%%%%%%%%%%%%%%%%%%%%%%%%%%%%%%%%%%%%%%%

We have presented a declarative framework to specify and enforce
purpose-aware policies.  In the literature, several proposals have
attempted to characterize the notion of purpose in the context of
security policies.  Some of them (e.g.,~\cite{ByunL08}) propose to
manage and enforce purpose by self-declaration, i.e.\ subjects
explicitly announce the purpose for accessing data.  While this
provides a first effort to embody purpose in access control policies,
these approaches are not able to precent malicious subjects from
claiming false purposes.
%%%%%%%%%%%%%%%%%%%%%%%%%%%%%%%%%%%%%%%%%%%%%%%%%%%%%%%%%%%%%%%%%%%%%%%%%%
%% ; this severely restricts the range of applicability to
%% non-sensitive private data processing environment
%%%%%%%%%%%%%%%%%%%%%%%%%%%%%%%%%%%%%%%%%%%%%%%%%%%%%%%%%%%%%%%%%%%%%%%%%%
Other works (e.g.,~\cite{qun-etal}) propose to extend the Role Based
Access Control model with mechanisms to automatically determine the
purpose for which certain data are accessed based on the roles of
subjects.  The main drawback of these approaches is the fact that
roles and purposes are not always aligned and members of the same role
may serve different purposes in their actions.  Other approaches
(e.g.,~\cite{ardagna-etal}) are based on extensions of (Attribute
Based) Access Control models for handling personal data in web
services or (e.g.,~\cite{purpose-distributed}) on extensions of Usage
Control models for the distributed enforcement of the purpose for data
usage (see, e.g.,~\cite{pretschner-etal}).  The main problem of these
approaches derives from the limited capability of the application
initiating the handling of personal data to control it when this has
been transferred to another (remote) application in the distributed
system.  Our framework avoids this problem by assuming a central
(workflow-based) system acting as a Policy Enforcement Point (PEP)
intercepting all requests of executing a particular action on certain
data.  Such an architecture for the PEP is reasonable in cloud-based
environments (e.g., the Smart Campus platform briefly described in
Section~\ref{sec:motivations}) in which services and applications are
implemented via of Application Programming Interfaces (APIs) so that
remote applications provide only the user interface while the control
logic is executed on the cloud platform and can thus be under the
control of the central PEP.

The common and more serious drawback of the approaches considered
above is that they fail to recognize that the purpose of an action (or
task) is determined by its position in a workflow, i.e.\ by its
relationships with other, interrelated, actions.  This observation has
been done in more recent
works---e.g.,~\cite{tschantz-etal,PetkovicPZ11,JafariEtal:ACM2014}---and
is also the starting point of our framework in this paper.  We share
the effort of formalizing the notion of purpose as a pre-requisite of
future actions.  In future work, we would like to study how
hierarchical workflows (i.e.\ workflows containing complex activities,
specified in terms of lower level tasks) can be expressed in our
framework so as to capture the specification of purposes as high-level
activities as done in~\cite{JafariEtal:ACM2014}.  The main difference
with previous works is our focus on run-time enforcement of
purpose-aware policies while~\cite{tschantz-etal}
and~\cite{PetkovicPZ11} on auditing.  Furthermore, the formal
framework adopted to develop these proposals are different from ours:
\cite{tschantz-etal} uses Markov Decision Processes,
\cite{PetkovicPZ11} a process calculus, and~\cite{JafariEtal:ACM2014}
an \emph{ad hoc} modal logic.  These choices force the authors
of~\cite{PetkovicPZ11,JafariEtal:ACM2014} to design algorithms for
policy enforcement or auditing from scratch.  For instance,
\cite{JafariEtal:ACM2014} gives a model checking algorithm on top of
which a run-time monitor for purpose-based policies can be
implemented, without studying its complexity.  In contrast, we use a
first-order temporal logic which comes with a wide range of techniques
to solve logical problems (see, e.g.,~\cite{DeGiacomoMGMM14}) that can
be reused (or adapted) to support the run-time enforcement of
purpose-aware policies.  For instance, we were able to derive the
first (to the best of our knowledge) complexity result of answering
authorization requests at run-time under a given purpose-aware policy:
exponential time in the number of authorization constraints
(Theorem~\ref{thm:main}).  This complexity is somehow intrisic to the
problem when assuming that the purpose of an action is determined by
its position in a workflow, as we do in this paper.  As a consequence,
achieving a purpose amounts to the succesful execution of the
associated workflow, which is called the Workflow Satisfiability
Problem (\wsp) in the literature~\cite{crampton} and is known to be
NP-hard already in the presence of one \sod constraint~\cite{wang-li}.
Several works have proposed techniques to solve both the
off-line~\cite{wang-li,YangXieRayLu} and the
on-line~\cite{basin2012,dossantos2015} version of the \wsp but none
considers purpose-aware policies as we do here.

Indeed, we are not the first to use first-order \LTL for the
specification of security policies.  For example, \cite{ltl-rbac}
shows how to express various types of \sod constraints in the Role
Based Access Control model by using first-order \LTL.  However, the
paper provides no method for the run-time monitoring of such
constraints and do not discuss if and how the approach can be used in
the context of workflow specifications.  Instead, the work
in~\cite{sttt-crampton} uses (a fragment of propositional) \LTL to
develop algorithms for checking the successfull termination of a
workflow.  Both works do not discuss purpose-aware policies.  To the
best of our knowledge, only the approach in~\cite{barth-etal} shares
with ours the use of first-order \LTL to specify and enforce utility
and privacy in business workflows.  The main difference is in the
long-term goal: we want to give rigorous foundation to specification
and verification techniques for purpose-based policies,
while~\cite{barth-etal} is seen as a first step towards to the
development of a general, clear, and comprehensive framework for
reducing high-level utility and privacy requirements to specific
operating guidelines that can be applied at individual steps in
business workflows.  It would be an interesting future work to see if
and how the two approaches can be combined together in order to derive
purpose-aware policies from high-level privacy requirements typically
found in laws, directives, and regulations.

Our choice of using a first-order \LTL formula
(Section~\ref{subsec:semantics}) as the semantics of purpose has two
main advantages.  First, it allows us to reuse well-known techniques
to specify access control policies, what we call data- and
rule-centric policies, by using (fragments of) first-order logic (see,
e.g.,~\cite{constraintdatalog,arkoudas-etal}).  Second, it allows us
to reuse the techniques for the specification and enforcement of
workflows put forward by the declarative approach to business process
specification in~\cite{AalstPS09,WestergaardM11,MaggiMWA11}.  However,
these works focus on tasks and their execution constraints,
disregarding the security and privacy aspects related to accessing the
data manipulated by the tasks.  A first proposal of adding the data
dimension to this approach is in~\cite{DeMasellisMM14}, which has been
from which we have borrowed the construction of the automata for
off-line and on-line verification in this paper
(Section~\ref{sec:verification}).  The choice of considering
first-order \LTL over finite instead of infinite traces goes back
to~\cite{DeGiacomoDMM14} in which it is argued that this is the right
choice for business process which are supposed to terminate as the
workflows associated purposes, which can be achieved in this way.  In
general, monitoring first-order \LTL formulae is
undecidable~\cite{BauerKV13} but, under the finite domain assumption,
\cite{DeMasellisS13} shows decidability.  Such an assumption is
reasonable in our framework where subjects are usually employees of a
company (e.g., the job hunting organization in
Section~\ref{sec:motivations}) whose number is bounded.  Our
verification techniques are also related to mechanisms for enforcing
security policies, such
as~\cite{pretschner-etal,basin-monitoring-sec,schneider}.  However,
these works mainly focus on access or usage control policies and are
of limited or no use for purpose-based policies considered in this
paper.

\end{sloppypar}

\bibliographystyle{splncs03}

\begin{small}
\bibliography{main-bib}

\begin{thebibliography}{10}
\providecommand{\url}[1]{\texttt{#1}}
\providecommand{\urlprefix}{URL }

\bibitem{EU-Directive95/46/EC}
Directive 95/46/ec of the european parliament and of the council of 24 october,
  1995,
  \url{http://eur-lex.europa.eu/LexUriServ/LexUriServ.do?uri=CELEX:31995L0046:en:HTML}

\bibitem{AalstPS09}
van~der Aalst, W.M.P., Pesic, M., Schonenberg, H.: Declarative workflows:
  Balancing between flexibility and support. Computer Science - R{\&}D  23(2),
  99--113 (2009)

\bibitem{arkoudas-etal}
Arkoudas, K., Chadha, R., Chiang, C.J.: Sophisticated access control via {SMT}
  and logical frameworks. Proc. of {ACM} TISSEC  16(4), ~17 (2014)

\bibitem{barth-etal}
Barth, A., Datta, A., Mitchell, J.C., Sundaram, S.: Privacy and utility in
  business processes. In: Proc. of 20th IEEE Computer Security Foundations
  Symposium (July 2007)

\bibitem{basin-monitoring-sec}
Basin, D., Klaedtke, F., M\"{u}ller, S.: Monitoring security policies with
  metric first-order temporal logic. In: Proc. of ACM SACMAT. pp. 23--34.
  SACMAT, ACM, New York, NY, USA (2010)

\bibitem{basin2012}
Basin, D., Burri, S.J., Karjoth, G.: Dynamic enforcement of abstract separation
  of duty constraints. ACM TISSeC  15(3),  13:1--13:30 (Nov 2012)

\bibitem{BauerKV13}
Bauer, A., K{\"{u}}ster, J., Vegliach, G.: From propositional to first-order
  monitoring. In: Runtime Verification. pp. 59--75 (2013)

\bibitem{dossantos2015}
Bertolissi, C., dos Santos, D.R., Ranise, S.: {Automated Synthesis of Run-time
  Monitors to Enforce Authorization Policies in Business Processes}. In: Asia
  CCS. ACM (2015)

\bibitem{byun-sacmat2005}
Byun, J.W., Bertino, E., Li, N.: Purpose based access control of complex data
  for privacy protection. In: Proc. of the ACM SACMAT. pp. 102--110. ACM, New
  York, NY, USA (2005)

\bibitem{ByunL08}
Byun, J., Li, N.: Purpose based access control for privacy protection in
  relational database systems. {VLDB} J.  17(4),  603--619 (2008)

\bibitem{ardagna-etal}
C.A.Ardagna, M.Cremonini, S.{De Capitani di Vimercati}, P.Samarati: A
  privacy-aware access control system. Journal of Computer Security (JCS)
  16(4),  369--392 (September 2008)

\bibitem{crampton}
Crampton, J.: A reference monitor for workflow systems with constrained task
  execution. In: Proc. of ACM SACMAT. pp. 38--47. ACM (2005)

\bibitem{sttt-crampton}
Crampton, J., Huth, M., Kuo, J.P.: Authorized workflow schemas: deciding
  realizability through $\mathsf{LTL}(\mathsf{F})$ model checking. Int. J. on
  Software Tools for Technology Transfer (STTT)  16(1),  31--48 (2014)

\bibitem{samarati-vimercati}
{De Capitani di Vimercati}, S., Foresti, S., Jajodia, S., Samarati, P.: Access
  control policies and languages. IJCSE  3(2),  94--102 (2007)

\bibitem{DeGiacomoMGMM14}
{De Giacomo}, G., {De Masellis}, R., Grasso, M., Maggi, F.M., Montali, M.:
  Monitoring business metaconstraints based on {LTL} and {LDL} for finite
  traces. In: Proc. of BPM. pp. 1--17 (2014)

\bibitem{DeGiacomoDMM14}
{De Giacomo}, G., {De Masellis}, R., Montali, M.: Reasoning on {LTL} on finite
  traces: Insensitivity to infiniteness. In: Proc. of {AAAI} Conference on
  Artificial Intelligence. pp. 1027--1033 (2014)

\bibitem{DeMasellisMM14}
{De Masellis}, R., Maggi, F.M., Montali, M.: Monitoring data-aware business
  constraints with finite state automata. In: Proc. of {ICSSP}. pp. 134--143
  (2014)

\bibitem{DeMasellisS13}
{De Masellis}, R., Su, J.: Runtime enforcement of first-order {LTL} properties
  on data-aware business processes. In: Proc. of {ICSOC}. pp. 54--68 (2013)

\bibitem{wpes09}
Jafari, M., Safavi{-}Naini, R., Sheppard, N.P.: Enforcing purpose of use via
  workflows. In: Proc. of {WPES}. pp. 113--116 (2009)

\bibitem{JafariEtal:ACM2014}
Jafari, M., Safavi-Naini, R., Fong, P.W.L., Barker, K.: A framework for
  expressing and enforcing purpose-based privacy policies. ACM Trans. Inf.
  Syst. Secur.  17(1),  3:1--3:31 (Aug 2014)

\bibitem{kroger-merz}
Kr{\"o}ger, F., Merz, S.: {Temporal Logic and State Systems}. Texts in
  Theoretical Computer Science. An EATCS Series, {Springer} (2008)

\bibitem{constraintdatalog}
Li, N., Mitchell, J.C.: {Datalog with constraints: a foundation for trust
  management languages}. In: PADL'03. pp. 58--73 (2003)

\bibitem{MaggiMWA11}
Maggi, F.M., Montali, M., Westergaard, M., van~der Aalst, W.M.P.: Monitoring
  business constraints with linear temporal logic: An approach based on colored
  automata. In: {BPM}. pp. 132--147 (2011)

\bibitem{purbac}
Masoumzadeh, A., Joshi, J.B.: {PuRBAC: Purpose-Aware Role-Based Access
  Control}. In: On the Move to Meaningful Internet Systems: OTM 2008, LNCS,
  vol. 5332, pp. 1104--1121. Springer Berlin Heidelberg (2008)

\bibitem{ltl-rbac}
Mossakowski, T., Drouineaud, M., Sohr, K.: A temporal-logic extension of
  role-based access control covering dynamic separation of duties. In: Proc. of
  TIME-ICTL. pp. 83--90 (2003)

\bibitem{YangXieRayLu}
P.~Yang, X.~Xie, I.R., Lu, S.: Satisfiability analysis of workflows with
  control-flow patterns and authorization constraints. IEEE TSC  99 (2013)

\bibitem{PetkovicPZ11}
Petkovic, M., Prandi, D., Zannone, N.: Purpose control: Did you process the
  data for the intended purpose? In: Secure Data Management. pp. 145--168
  (2011)

\bibitem{pretschner-etal}
Pretschner, A., Hilty, M., Basin, D.: {Distributed usage control}. Comm. ACM
  49,  39--44 (2006)

\bibitem{qun-etal}
Qun, N., Elisa, B., Jorge, L., Carolyn, B., Karat, C.M., Alberto, T.:
  {Privacy-aware Role-Based Access Control}. TISSeC  13,  1--31 (July 2010)

\bibitem{purpose-distributed}
Rath, A.T., Colin, J.N.: {Modeling and Expressing Purpose Validation Policy for
  Privacy-aware Usage Control in Distributed Environment}. In: Proc. of ICUIMC.
  pp. 14:1--14:8. ACM, New York, NY, USA (2014)

\bibitem{schneider}
Schneider, F.B.: {Enforceable security policies}. TISSeC  3,  30--50 (2000)

\bibitem{tschantz-etal}
Tschantz, M.C., Datta, A., Wing, J.M.: Formalizing and enforcing purpose
  restrictions in privacy policies. In: IEEE Symposium on Security and Privacy.
  pp. 176--190. IEEE Computer Society (2012)

\bibitem{wang-li}
Wang, Q., Li, N.: Satisfiability and resiliency in workflow authorization
  systems. TISSeC  13,  40:1--40:35 (December 2010)

\bibitem{WestergaardM11}
Westergaard, M., Maggi, F.M.: Declare: {A} tool suite for declarative workflow
  modeling and enactment. In: Proc. of BPM (2011)

\bibitem{westin}
Westin, A.: {Privacy and Freedom}. Atheneum, New York, USA (1968)

\end{thebibliography}
\end{small}

%\appendix
%\input{appendix}

\end{document}